%% file: 0_main.tex
\documentclass[11pt, oneside]{article}
\input{_prelude}

\input{_prelude_2}

\input{_authors}

\title{Geometry of Rounding: Near Optimal Bounds and a New Neighborhood Sperner's Lemma
\thanks{Research supported in part by NSF grant 1934884, 2130536, 2130608}
}

\date{\today}

\begin{document}
\maketitle

\renewcommand{\P}{\mathcal{P}}

\begin{abstract}
A partition $\P$ of $\R^d$ is called a $(k,\varepsilon)$-secluded partition if, for every $\vec{p} \in \R^d$, the ball $\ballcinf{\varepsilon}{\vec{p}}$ intersects at most $k$ members of $\P$. A goal in designing such secluded partitions is to minimize $k$ while making $\varepsilon$ as large as possible. This partition problem has connections to a diverse range of topics, including deterministic rounding schemes, pseudodeterminism, replicability, as well as Sperner/KKM-type results.

In this work, we establish near-optimal relationships between $k$ and $\varepsilon$. We show that, for any bounded measure partitions and for any $d\geq 1$, it must be that $k\geq(1+2\varepsilon)^d$. Thus, when $k=k(d)$ is restricted to ${\rm poly}(d)$, it follows that $\varepsilon=\varepsilon(d)\in O\left(\frac{\ln d}{d}\right)$. This bound is tight up to log factors, as it is known that there exist secluded partitions with $k(d)=d+1$ and $\epsilon(d)=\frac{1}{2d}$. We also provide new constructions of secluded partitions that work for a broad spectrum of $k(d)$ and $\varepsilon(d)$ parameters. Specifically, we prove that, for any $f:\mathbb{N}\rightarrow\mathbb{N}$, there is a secluded partition with $k(d)=(f(d)+1)^{\ceil{\frac{d}{f(d)}}}$ and $\varepsilon(d)=\frac{1}{2f(d)}$. These new partitions are optimal up to $O(\log d)$ factors for various choices of $k(d)$ and $\varepsilon(d)$. Based on the lower bound result, we establish a new neighborhood version of Sperner's lemma over hypercubes, which is of independent interest. In addition, we prove a no-free-lunch theorem about the limitations of rounding schemes in the context of pseudodeterministic/replicable algorithms.

We use techniques from measure theory and the geometry of numbers, including the generalized Brunn-Minkowski inequality, the isodiametric inequality, and Blichfeldt's theorem, to establish our results.

\end{abstract}

\pagenumbering{roman}
\newpage
\tableofcontents

\newpage
\pagenumbering{arabic}

\input{1_introduction}

\input{2_proof_outlines}

\input{3_notation}
\input{4_general_bounds}

\input{5_sperner_kkm_lebesgue}
\input{6_product_partitions}

\input{7_no_free_lunch}
\newpage
\bibliography{references,ref}
\bibliographystyle{alpha} 
\appendix
\input{appendix_measure_theory}
\input{appendix_minkowski_sums}

\input{appendix_other}
\end{document}

%% file: _prelude.tex
\input{_thmtools_setup}

\usepackage[dvipsnames]{xcolor}
\usepackage{amsmath}
\usepackage{amsthm}
\usepackage{amssymb}
\usepackage{changepage} 
\usepackage[mathscr]{euscript}
\usepackage{enumitem}
\setlist{nolistsep} 
\usepackage{verbatim}
\usepackage{color}
\usepackage{geometry}
\usepackage{graphicx}
\usepackage{ifthen}
\usepackage{mathtools}
\usepackage{mathdots}
\usepackage{lipsum}
\usepackage[many]{tcolorbox}
\usepackage{complexity}
\usepackage{hyperref}
\usepackage{tikz-cd} 
\usepackage{setspace}
\usepackage{algorithm}
\usepackage[noend]{algpseudocode}
\usepackage{thmtools, thm-restate}
\usepackage{todonotes}
\usepackage{subcaption}
\usepackage{faktor}

\newcounter{mycomment}

\newcounter{proposed}

\newboolean{print} 
\setboolean{print}{true}

\newfunc{\exponential}{exp}

\newlang{\langINDEX}{INDEX}
\newlang{\langLENGTHGEQ}{LENGTHGEQ}
\newlang{\langtemp}{A}          
\newlang{\langREACH}{REACH}

\newcommand{\set}[1]{\left\{#1\right\}}

\newcommand{\closure}{\overline}
\newcommand{\clos}{\overline}

\newcommand{\N}{\mathbb{N}}     
\renewcommand{\R}{\mathbb{R}}   
\renewcommand{\C}{\mathbb{C}}   

\renewcommand{\S}{\mathscr{S}}
\renewcommand{\P}{\mathcal{P}}

\DeclareMathOperator{\range}{range}
\DeclareMathOperator{\defeq}{\overset{def}{=}}

\DeclareMathOperator{\diam}{diam}

\DeclarePairedDelimiter{\ceil}{\lceil}{\rceil}
\DeclarePairedDelimiter\floor{\lfloor}{\rfloor}

\DeclarePairedDelimiter\abs{\lvert}{\rvert}

\DeclarePairedDelimiter\norm{\lVert}{\rVert}

\newcommand{\mytcbtheoremconstructor}[4]{

  \ifthenelse{\boolean{print}}{
    \newtcbtheorem[use counter=boxcounter, number within=section]{#1}{#2}{
      boxrule=0.7mm,
      skin=enhanced,
      middle=3mm,
      breakable, 
      colback=white,
      colframe=#4!50!black,
      fonttitle=\bfseries
    }{#3}
  }{
    \newtcbtheorem[use counter=boxcounter, number within=section]{#1}{#2}{
      boxrule=0.7mm,
      skin=enhanced,
      middle=3mm,
      breakable, 
      bicolor,
      colback=#4!20,
      colbacklower=white,
      colframe=#4!50!black,
      fonttitle=\bfseries,
    }{#3}
  }

}

%% file: _thmtools_setup.tex
\usepackage{amsthm}
\usepackage{amsmath}

\usepackage{letltxmacro}                         
\LetLtxMacro\amsproof\proof                      
\LetLtxMacro\amsendproof\endproof                
\usepackage{thmtools}                            %
\AtBeginDocument{
  \LetLtxMacro\proof\amsproof                    %
  \LetLtxMacro\endproof\amsendproof              %
}                                                %

\usepackage{thm-restate}
\usepackage[colorlinks=true, allcolors=NavyBlue]{hyperref}
\usepackage{
  hyperref,
  cleveref,
}
\usepackage[dvipsnames]{xcolor} 

\counterwithin*{COUNTERHACKSUB}{COUNTERHACK}

\newcommand{\MyThmtoolsConstructor}[3]{
  \declaretheorem[
  numberlike=COUNTERHACK,
  numbered=yes,
  name=#2,
  Refname={#2}, 
  refname={\MakeLowercase{#2}},  
  shaded={bgcolor=#3, margin=0pt, padding=5pt},
  ]{#1}
  \declaretheorem[
  numbered=no,
  name=Informal #2,
  Refname={#2}, 
  refname={\MakeLowercase{#2}},  
  shaded={bgcolor=#3, margin=0pt, padding=5pt},
  ]{#1-informal}
  \declaretheorem[
  numbered=no,
  name=#2,
  Refname={#2}, 
  refname={\MakeLowercase{#2}},  
  ]{#1*}
}

\newcommand{\MySubThmtoolsConstructor}[3]{
  \declaretheorem[
  numberlike=COUNTERHACKSUB,
  numbered=yes,
  name=#2,
  Refname={#2}, 
  refname={\MakeLowercase{#2}},  
  ]{#1}
  \declaretheorem[
  numbered=no,
  name=Informal #2,
  Refname={#2}, 
  refname={\MakeLowercase{#2}},  
  ]{#1-informal}
  \declaretheorem[
  numbered=no,
  name=#2,
  Refname={#2}, 
  refname={\MakeLowercase{#2}},  
  ]{#1*}
}

\MyThmtoolsConstructor{definition}   {Definition}     {Gray!20}
\MyThmtoolsConstructor{theorem}      {Theorem}        {Gray!20}
\MyThmtoolsConstructor{corollary}    {Corollary}      {Gray!20}
\MyThmtoolsConstructor{porism}       {Porism}         {Gray!20}
\MyThmtoolsConstructor{lemma}        {Lemma}          {Gray!20}
\MyThmtoolsConstructor{proposition}  {Proposition}    {Gray!20}
\MyThmtoolsConstructor{example}      {Example}        {Gray!20}
\MyThmtoolsConstructor{comment}      {Comment}        {Gray!20}
\MyThmtoolsConstructor{exercise}     {Exercise}       {Gray!20}
\MyThmtoolsConstructor{conjecture}   {Conjecture}     {Gray!20}
\MyThmtoolsConstructor{reflection}   {Reflection}     {Gray!20}
\MyThmtoolsConstructor{question}     {Question}       {Gray!20}
\MyThmtoolsConstructor{observation}  {Observation}    {Gray!20}
\MyThmtoolsConstructor{claim}        {Claim}          {Gray!20}
\MyThmtoolsConstructor{fact}         {Fact}           {Gray!20}
\MyThmtoolsConstructor{note}         {Note}           {Gray!20}
\MyThmtoolsConstructor{remark}       {Remark}         {Gray!20}
\MyThmtoolsConstructor{notation}     {Notation}       {Gray!20}
\MyThmtoolsConstructor{summary}      {Summary Result} {Gray!20}
\MyThmtoolsConstructor{problem}      {Problem}        {Gray!20}

\MySubThmtoolsConstructor{subclaim}     {Claim}          {Gray!20}

\newcommand*{\subclaimproofname}{Proof of Claim}
\newenvironment{subclaimproof}[1][\subclaimproofname]{\begin{proof}[#1]}{\end{proof}}

%% file: _prelude_2.tex
\newcommand{\unitballmeasurestd}{v_{{\scriptscriptstyle \norm{\cdot}},d}}
\newcommand{\unitballmeasureinf}{v_{{\scriptscriptstyle \norm{\cdot}_\infty},d}}

\newcommand{\diamstd}{\diam_{\norm{\cdot}}}
\newcommand{\diaminf}{\diam_{\infty}}

\newcommand{\ballo}[2]{B^{\circ}(#1,#2)}
\newcommand{\ballc}[2]{\clos{B}(#1,#2)}

\newcommand{\ballon}[3]{B^\circ_{#3}(#1,#2)}
\newcommand{\ballcn}[3]{\clos{B}_{#3}(#1,#2)}

\newcommand{\balloinf}[2]{\ballon{#1}{#2}{\infty}} 
\newcommand{\ballcinf}[2]{\ballcn{#1}{#2}{\infty}} 
\newcommand{\ballostd}[2]{\ballon{#1}{#2}{\scriptscriptstyle\norm{\cdot}}} 
\newcommand{\ballcstd}[2]{\ballcn{#1}{#2}{\scriptscriptstyle\norm{\cdot}}} 

\newcommand*{\numberfullref}[1]{\hyperref[{#1}]{\Autoref*{#1} (\nameref*{#1})}}
\newcommand*{\namefullref}[1]{\hyperref[{#1}]{\nameref*{#1} (\Autoref*{#1})}}

\newfunc{\subpoly}{subpoly}
\newfunc{\quasipoly}{quasipoly}
\newfunc{\strongsubexp}{strongsubexp}
\newfunc{\weaksubexp}{weaksubexp}
\newfunc{\eexp}{exp}

\renewcommand{\epsilon}{\varepsilon}

%% file: _authors.tex
\usepackage{authblk}

\makeatletter
\newcommand\email[2][]%
  {\newaffiltrue\let\AB@blk@and\AB@pand
      \if\relax#1\relax\def\AB@note{\AB@thenote}\else\def\AB@note{\relax}%
        \setcounter{Maxaffil}{0}\fi
      \begingroup
        \let\protect\@unexpandable@protect
        \def\thanks{\protect\thanks}\def\footnote{\protect\footnote}%
        \@temptokena=\expandafter{\AB@authors}%
        {\def\\{\protect\\\protect\Affilfont}\xdef\AB@temp{\url{#2}}}%
         \xdef\AB@authors{\the\@temptokena\AB@las\AB@au@str
         \protect\\[\affilsep]\protect\Affilfont\AB@temp}%
         \gdef\AB@las{}\gdef\AB@au@str{}%
        {\def\\{, \ignorespaces}\xdef\AB@temp{\url{#2}}}%
        \@temptokena=\expandafter{\AB@affillist}%
        \xdef\AB@affillist{\the\@temptokena \AB@affilsep
          \AB@affilnote{}\protect\Affilfont\AB@temp}%
      \endgroup
      \let\AB@affilsep\AB@affilsepx
}
\makeatother

\renewcommand\Affilfont{\itshape\small}
\author{Jason Vander Woude}
\affil{School of Computing and Department of Mathematics, University of Nebraska-Lincoln}
\email{jasonvw@huskers.unl.edu}

\author{Peter Dixon}
\affil{Ben-Gurion University of the Negev}
\email{tooplark@gmail.com}

\author{A. Pavan}
\affil{Department of Computer Science, Iowa State University}
\email{pavan@cs.iastate.edu}

\author{Jamie Radcliffe}
\affil{Department of Mathematics, University of Nebraska-Lincoln}
\email{aradcliffe1@unl.edu}

\author{ N.~V.~Vinodchandran}
\affil{School of Computing, University of Nebraska-Lincoln}
\email{vinod@cse.unl.edu}


%% file: 1_introduction.tex
\section{Introduction}
\label{sec:introduction}

We investigate the following secluded partition problem. We use the notation $\ballcstd{\epsilon}{\vec{p}}$ and $\ballostd{\epsilon}{\vec{p}}$ to indicate, respectively, the closed or open ball of radius $\epsilon$ centered at $\vec{p}$ for a general norm, and we will use $\ballcinf{\epsilon}{\vec{p}}$ and $\balloinf{\epsilon}{\vec{p}}$ when the norm is the $\ell_\infty$ norm.

\begin{definition}[\cite{geometry_of_rounding}]
A partition $\P$ of $\R^d$ is called a $(k,\varepsilon)$-secluded partition if for every $\vec{p} \in \R^d$, the ball $\ballcinf{\varepsilon}{\vec{p}}$ intersects at most $k$ members of $\P$. 
The parameters $k$ and $\varepsilon$ are called the {\em degree} and {\em tolerance} respectively.
\end{definition}

Ideally, we would like to construct partitions with degree ($k$) as small as possible and tolerance ($\varepsilon$) as large as possible. However, there is a natural trade-off between these two parameters.
It is easy to see that for the standard grid partition of $\R^d$ with half-open unit hypercubes, we can have $k=2^d$ and $\varepsilon=\frac12$.   Somewhat surprisingly, $k$ can be made exponentially smaller at the expense of making $\varepsilon$ (polynomially) small. In particular, it is known that there is a unit hypercube partition of $\R^d$ with $k = d+1$ and $\varepsilon = \frac{1}{2d}$~\cite{geometry_of_rounding}. Also, a certain rounding scheme given in \cite{hoza_preserving_2018} induces a $(d+1,\frac{1}{6(d+1)})$-secluded partition (with members which are not hypercubes). It is known that for $\varepsilon>0$, any $(k,\varepsilon)$-secluded partition of $\R^d$ with at most unit diameter members (in $\ell_\infty$) must have $k\geq d+1$. In addition, if $k=d+1$ then it must be that $\varepsilon\leq\frac{2}{\sqrt{d}}$~\cite{geometry_of_rounding}. These are the only results known about the construction/impossibility results of secluded partitions. 

The following questions naturally arise. (1)  
Can we design a secluded partition with degree $k=k(d)\in \poly(d)$ with tolerance $\varepsilon=\varepsilon(d)\in\omega(\frac1{d})$? (2) Does allowing $k(d)$ to be sub-exponential in $d$ allow for $\varepsilon(d)$ to be a constant? (3) More generally, and very fundamentally, what is the trade-off between the degree and the tolerance parameters in secluded partitions? In this work, we establish a near-optimal trade-off between these two parameters.

\subsection{Motivation}

Investigating constructions and trade-offs for  secluded partitions is a fundamental question in its own right. However, its applicability to various topics in theoretical computer science is also a significant motivation for investigation.

\medskip
\noindent{\bf\em Deterministic Rounding:}
Rounding schemes are fundamental in computer science with  applications in approximation algorithms, pseudorandomness, replicability, and differential privacy~\cite{RT87,SaksZhou99, DixonPavanVinod18,impagliazzo_reproducibility_2022,Goldreich19,GrossmanLiu19,kindler_spherical_2012,kindler_spherical_2008}. There is a natural equivalence between rounding schemes and partitions of Euclidean space in a very general sense that has been observed in prior works~\cite{kindler_spherical_2008,kindler_spherical_2012,geometry_of_rounding}. In particular, certain deterministic rounding schemes are equivalent to secluded partitions. 

\begin{definition}
A deterministic rounding scheme is a family of functions ${\mathcal F} = \{f_d\}_{d\in\N}$ where $f_d: \R^d \to \R^d$. $\mathcal{F}$ is called a  $(k(d),\varepsilon(d))$-deterministic rounding scheme if two properties hold for each $d\in\N$: (1) for all $\vec{x}\in\R^d$, $\norm{f_d(\vec{x})-\vec{x}}_\infty\leq \frac12$\footnote{
The value $\frac{1}{2}$ in the above definition is arbitrary and can be replaced with another constant by appropriately scaling other parameters.}, and (2) for all $\vec{p}\in\R^d$ the set $\set{f_d(\vec{x})\colon \vec{x}\in \ballcinf{\varepsilon(d)}{\vec{p}}}$ has cardinality at most $k(d)$.
\end{definition}

 Let $f: \R^d \rightarrow \R^d$ be a rounding function. For any $\vec{y} \in \R^d$, let $X_{\vec{y}}=\set{\vec{x}\in\R^d\colon f(\vec{x})=\vec{y}}$ denote the points rounded to $\vec{y}$. Let $\P_f = \set{X_{\vec{y}}\colon \vec{y}\in\range(f)}$ which is clearly a partition of $\R^d$. 
Conversely, a partition $\P$ of $\R^d$ induces (many) deterministic rounding functions $f_{\P}$ as follows: for each member $X \in \P$ let $\vec{p}_{X} \in \R^d$ (often $\vec{p}_X\in X$) be some fixed representative of the member $X$. Then the rounding function $f_{\P}$ maps any point $\vec{x} \in X$ to $\vec{p}_{X}$. This leads to the following observation~\cite{geometry_of_rounding}.

\begin{observation}[Equivalence of Rounding Schemes and Partitions]\label{:rounding-schemes-and-partitions}
A $(k(d), \epsilon(d))$-deterministic rounding scheme induces, for each $d\in\N$, a $(k(d), \epsilon(d))$-secluded partition of $\R^d$ in which each member has diameter at most $1$. Conversely, a sequence $\langle \P_d\rangle_{d=1}^\infty$ of partitions where $\P_d$ is $(k(d), \epsilon(d))$-secluded and contains only members of diameter at most $1$ induces many $(k(d), \epsilon(d))$-deterministic rounding schemes.
\end{observation}

The observation provides a geometric perspective on rounding schemes. By investigating secluded partitions, this work examines the concept of rounding from a geometric viewpoint and explores its possibilities and limitations. Rounding schemes and secluded partitions are related to the notion of pseudodeterminism which we discuss next.

\medskip
\noindent{\bf\em Pseudodeterminism and Replicability:}
Due to its unexpected links with computational complexity theory~\cite{OliveiraSanthanam17,DixonPavanVinod18,LuOlivieraSanthanam21,DPWV22}, differential privacy~\cite{BLM20,ALMM19,GKM21}, and reproducible learning~\cite{impagliazzo_reproducibility_2022}, the concept of pseudodeterminism and its various forms have been gaining increased attention in theoretical computer science. A probabilistic algorithm $M$ is {\em pseudodeterministic} if, for every input $x$, there is a canonical value $v_x$ such that $M(x)$ outputs $v_x$ with high probability. This notion can be extended to $k$-pseudodeterminism where for every $x$, there is a canonical set $S_x$ of size at most $k$ and the output of $M(x)$ belongs to $S_x$ with high probability. Pseudodeterminism naturally captures the notion of reproducibility/replicability in computations. The concept of pseudodeterminism and its extensions have been applied to learning and privacy, leading to the seminal discovery of the equivalence between private PAC learning and Littlestone dimension~\cite{BLM20,ALMM19}.

One method for creating a $k$-pseudodeterministic algorithm is through the use of rounding~\cite{Goldreich19,geometry_of_rounding,hoza_preserving_2018}. Let $A$ be a probabilistic algorithm that approximates a function $f$ whose range is $[0, 1]^d$ with an additive error $\nu$ in 
 $\ell_{\infty}$ norm. In general, the number of possible outcomes of such an approximation algorithm could be very large.
 However, by rounding each coordinate of $A(x)$ to the nearest multiple of $\nu$, we can create a $2^d$-pseudodeterministic algorithm with an error of $2\nu$. In general, we can use a $(k(d),\varepsilon(d))$-deterministic rounding scheme (equivalently by \Autoref{:rounding-schemes-and-partitions}, a unit diameter $(k(d),\varepsilon(d))$-secluded partition) to obtain a $k(d)$-pseudodeterministic algorithm with an approximation error of $O(\nu/\varepsilon)$. 
 Using the known $(d+1,{1\over 2d})$-secluded partition, this implies that any $\nu$-approximation algorithm can be converted to a $(d+1)$-pseudodeterministic algorithm with an {\em approximation error blown up to} $O(\nu d)$.  Thus, to achieve a final approximation error of $\varepsilon'$, the original approximation algorithm $A$
 must start with a smaller error of $O(\varepsilon'/d)$.  Since the running time of approximation algorithms typically depends on the desired quality of the approximation,  this leads to larger run times for approximation algorithms (or sample complexity in the context of learning). A critical question is whether this {\em linear blowup} in the approximation error is necessary. Can we convert $A$ into a $\poly(d)$-pseudodeterministic algorithm with only a smaller approximation-error blowup, say $d^{\gamma}$, for some $0 < \gamma < 1$? 
A trade-off result that we establish between the degree and the tolerance  for secluded partitions  leads to 
a {\em no-free-lunch} theorem: for any generic rounding scheme that produces a $k(d)$-pseudodeterministic algorithm, there must be a degradation in approximation quality by a factor of at least $\Omega(\frac{d}{\ln k(d)})$. Consequently, if $k(d) \in \poly(d)$, then the approximation accuracy must decrease by at least a factor of $\Omega(\frac{d}{\ln d})$. 

\medskip
\noindent{\bf\em Connection to Sperner's Lemma and the KKM Lemma:}
Sperner's lemma and the KKM lemma have found applications in theoretical computer science~\cite{Papadimitriou90,ChenD09, SaksZ00,HerlihyShavit99, DaskalakisGP09,BorowskyG93,CLLORS95,DPWV22}. Somewhat surprisingly to the authors, our investigation to establish the trade-off between the degree and the tolerance parameters of secluded partitions has led to the discovery of a new ``neighborhood" version of the Sperner/KKM lemma. This new lemma has already been useful in establishing new results on fair-division problems (ongoing work). We expect that the secluded partition problem and the results we establish will have further applications in other contexts where the Sperner/KKM lemma has been utilized.

To conclude, exploring the constructions and trade-offs associated with secluded partitions is an essential inquiry in its own regard.  The fact that it has implications across a variety of subjects within theoretical computer science is an additional driving force for further investigation.

\subsection{Our Contributions}

Our work makes four key contributions. Firstly, we establish a lower bound on the degree parameter $k$ in terms of the tolerance parameter $\varepsilon$, for any $(k,\varepsilon)$-secluded partition. This result is established in a very general setting and works for {\em any} norm in Euclidean space. Secondly, using the techniques  developed in the proof of the degree lower bound result, we establish a new ``neighborhood" variant of the Sperner/KKM lemma that is of independent interest. Thirdly, we give a new generic construction of $(k, \varepsilon)$-secluded partitions.  This construction together with the degree lower bound result establishes near optimality of the tolerance parameter  for various natural choices of degree parameter. Finally, we establish a no-free-lunch theorem for multi-pseudodeterministic computations: for any generic rounding method that produces a $k(d)$-pseudodeterministic algorithm for a function whose range is $\R^d$, there must be an almost linear degradation in approximation quality. 

\subsubsection{A Lower Bound on the Degree}

We state the result in terms of bounded-measure partitions and will discuss later why we use measure instead of diameter. For now, it is enough to note that for the $\ell_\infty$ norm, unit diameter implies unit outer measure (see \Autoref{diameter-ball} with $D=1$ and $M=1$), so any results stated for partitions with members of outer measure at most $1$ immediately implies the same result for partitions with members of $\ell_\infty$ diameter at most $1$.

\begin{restatable}{theorem}{restatableKExponentialInDimensionMeasure}\label{k-exponential-in-dimension-measure}
Let $d\in\N$, $\varepsilon\in[0,\infty)$, and $\P$ a partition of $\R^d$ such that every member has outer Lebesgue measure at most $1$. Then there exists some $\vec{p}\in\R^d$ such that $\balloinf{\varepsilon}{\vec{p}}$ intersects at least $(1+2\varepsilon)^d$ members of $\P$. Thus, if $\P$ is a $(k,\varepsilon)$-secluded partition, then $k\geq (1+2\varepsilon)^d$. Consequently, if $k\leq2^d$, then it must be that $\varepsilon\leq \frac{\ln k}{d}$. 
\end{restatable}

Thus even if $\varepsilon(d)$ is an arbitrarily small constant, the degree $k(d)$ has to be made exponential in $d$.  Also, it follows that for $k(d)$ to be $\poly(d)$, we must have $\varepsilon(d) \in O(\frac{\ln d}{d})$. This establishes that the $(d+1, \frac{1}{2d})$ partitions form~\cite{geometry_of_rounding} optimal up to a log factor. This also implies that  relaxing $k(d)=d+1$ to $k(d)\in\poly(d)$ can at best contribute $O(\ln d)$ factor improvements in $\varepsilon(d)$.

Our proof of \Autoref{k-exponential-in-dimension-measure} does not rely on any specific properties of the $\ell_{\infty}$ norm. In particular,
we establish the  following result   that holds for {\em any norm} in $\R^d$, and \Autoref{k-exponential-in-dimension-measure} follows from this as a corollary.

\begin{restatable}[$\varepsilon$-Neighborhoods for Measure Bounded Partitions and Arbitrary Norm]{theorem}{restatableMeasureBoundedAnyNorm}\label{epsilon-measure-bounded-any-norm}
Let $d\in\N$, $M\in(0,\infty)$, and $\P$ a partition of $\R^d$ such that every member has outer Lebesgue measure at most $M$. Let $\R^d$ be equipped with any norm $\norm{\cdot}$. For every $\varepsilon\in(0,\infty)$, there exists $\vec{p}\in\R^d$ such that $\ballostd{\epsilon}{\vec{p}}$ intersects at least $k=\ceil*{\left(1+\varepsilon\left(\frac{\unitballmeasurestd}{M}\right)^{1/d}\right)^d}$ members of the partition
where $\unitballmeasurestd\defeq m\left(\ballostd{1}{\vec{0}}\right)$ is the measure of the $\norm{\cdot}$ unit ball.
\end{restatable}

For the $\ell_{\infty}$ norm, $\unitballmeasurestd$ is $2^d$. Thus, \Autoref{k-exponential-in-dimension-measure} is a corollary of the above theorem with the $\ell_\infty$ norm and $M=1$.
Using the isodiametric inequality, we obtain the following diameter version as a corollary.

\begin{restatable}[$\varepsilon$-Neighborhoods for Diameter Bounded Partitions]{corollary}{restatableEpsilonDiameterBoundedAnyNorm}\label{epsilon-diameter-bounded-any-norm}
Let $d\in\N$, $D\in(0,\infty)$, $\varepsilon\in(0,\infty)$. Let $\R^d$ be equipped with any norm $\norm{\cdot}$. Let $\P$ be a partition of $\R^d$ such that every member has diameter at most $D$ (with respect to $\norm{\cdot}$). Then there exists $\vec{p}\in\R^d$ such that $\ballostd{\epsilon}{\vec{p}}$ intersects at least $k=\ceil*{\left(1+\frac{2\varepsilon}{D}\right)^d}$ members of the partition.
\end{restatable}

The expression $\frac{2\varepsilon}{D}$ in the above result should be viewed as a normalization factor which is the ratio of the diameter of the $\varepsilon$-ball to the diameter of the members of the partition. 
The expression $\epsilon\left(\frac{v_d}{M}\right)^{1/d}=\frac{(\epsilon^d\cdot v_d)^{1/d}}{M^{1/d}}$ should also viewed as a normalization factor. It is typical in measure theory contexts to see the $d$th roots of measures in $d$ dimensions show up, and they often serve as a type of characteristic length scale and are sometimes more robust than actual distances. Thus, this expression should be viewed as the ratio of the characteristic length scale of the $\epsilon$ ball (which has measure $\epsilon^d\cdot v_d$) to the characteristic length scale of the members of the partition (which have measure at most $M$).

\medskip
\noindent{\bf\em A Remark about Measure:}  Though rounding schemes are equivalent to partitions with bounded {\em diameter}, we have stated the above theorems with respect to partitions with bounded {\em outer measure}. As mentioned previously, in the case of the $\ell_\infty$ norm, this is a strict generalization. Furthermore, the measure perspective allows us to employ powerful tools from measure theory and the geometry of numbers such as Blichfeldt's theorem and the \namefullref{generalized-brunn-minkowski-inequality} for establishing these results. This also allows us to establish the lower bound result with respect to any norm.  Finally, we can convert measure results to diameter results very tightly using the \namefullref{isodiametric-ineq}.

\subsubsection{A Neighborhood Sperner/KKM/Lebesgue Theorem}

We use the techniques developed in establishing  \Autoref{epsilon-measure-bounded-any-norm}  to prove a new  variant Sperner's lemma on the cube.
The version of our result bears more resemblance to the KKM lemma---known to be equivalent to Sperner's lemma in a natural way. The KKM  Leamma states that for a covering/coloring of the $d$-simplex satisfying certain properties, there is some point in the closure of at least $d+1$ sets/colors. There is also a version of the KKM lemma for the cube \cite{kuhn66, komiya_simple_1994, wolsey_cubical_1977} which can be proven either independently or as consequences of versions of Sperner's lemma for the cube \cite{de_loera_polytopal_2002}. Though it is a lesser-known result, the Lebesgue covering theorem (c.f. \cite[Theorem~IV~2]{dimension_theory_lebesgue_covering}) says the same thing for the cube. These results are summarized by the following theorem.

\begin{theorem}[Sperner/KKM/Lebesgue]
Given a coloring of $[0,1]^d$ by finitely many colors in which no color includes points of opposite faces, there exists a point at the closure of at least $d+1$ different colors.
\end{theorem}

In some versions of the above lemma, there are exactly $2^d$ colors (one for each corner/vertex) and in some others, the coloring can assign multiple colors to a single point so that it is a covering rather than a partitioning by colors, but both of these differences are inconsequential. The finiteness is required to guarantee a point in the closure, but even without the finiteness, there exists a point such that for any $\epsilon>0$ the $\epsilon$ ball (in any norm) intersects at least $d+1$ colors. In this spirit, our new variant of this Sperner/KKM/Lebesgue result considers not a point but instead a small $\epsilon$ ball in the $\ell_\infty$ norm. Thus we call it the {\em Neighborhood Sperner Lemma}. 

\begin{restatable}[Neighborhood Sperner lemma]{theorem}{restatableNewSpernerKKMLebesgue}\label{new-sperner-kkm-lebesgue}
Given a coloring of $[0,1]^d$ in which no color includes points of opposite faces, then for any $\epsilon\in(0,\frac12]$ there exists a point $\vec{p}\in[0,1]^d$ such that $\balloinf{\epsilon}{\vec{p}}$ contains points of at least $\left(1+\frac23 \epsilon\right)^d$ different colors.
\end{restatable}

\subsubsection{New Constructions}

Currently, we know of only two types of secluded partitions---the standard grid partition which is $(2^d, \frac{1}{2})$-secluded and the  two different $(d+1, O(\frac{1}{d}))$-secluded partitions of \cite{geometry_of_rounding, hoza_preserving_2018}. We do not know of partitions with other values of $k(d)$ and $\varepsilon(d)$. For example, no constructions are known say for the parameter $\varepsilon(d) = {1\over \ln d}$ or for the parameter $k(d)\in\poly(d)$ or $k(d)$ sub-exponential in $d$. The  grid partition is optimal in the tolerance parameter and the $(d+1, O(\frac{1}{2d}))$ partitions are optimal in the degree parameter. By using these partitions as building blocks, we provide new constructions of secluded partitions that work for a range of degree and tolerance parameters.
We establish the following generic construction result:

\begin{restatable}{theorem}{restatableGlueConstruction}\label{basic-reclusive-gluing}
Let $f:\N\to\N$ be any function.
For each $d\in\N$, there exists a $(k(d),\varepsilon(d))$-secluded unit cube partition of $\R^d$ where $k(d)=\left(f(d)+1\right)^{\ceil{\frac{d}{f(d)}}}$ and $\varepsilon(d)=\frac{1}{2f(d)}$
\end{restatable}

Using this theorem we obtain secluded partitions  for various choices of $k(d)$ and $\varepsilon(d)$. For example, for any $\varepsilon(d) \in o(1)$, we can achieve $k(d) \in 2^{o(d)}$. For any $\varepsilon(d) \in O(\frac{1}{d})$ (even if $\epsilon(d)$ is larger that the values allows by \cite{geometry_of_rounding, hoza_preserving_2018}) we can still obtain $k(d)\in \poly(d)$. For $\varepsilon(d) = \frac{\ln^{\ell} d}{d}$, we get that $k(d) = 2^{\log^{\ell+1} d}$. 

These constructions are {\em near optimal}. For instance when $k(d)\in\poly(d)$, \Autoref{k-exponential-in-dimension-measure} implies that $\epsilon(d)\in O(\frac{\ln d}{d})$, and when $k = 2^{\log^{l+1} d}$, it must be that $\varepsilon(d)\in O(\frac{\ln^{l+1}d }{d})$. Thus, the value of $\varepsilon$ achieved by the construction is optimal up to an $O(\ln d)$ factor.

\subsubsection{A No-Free-Lunch Theorem}
Secluded partitions (deterministic rounding schemes) are used as a generic tool to design $k$-pseudodeterministic algorithms.  However, as discussed in the introduction,  this results in a linear  (in $d$) blowup of the approximation error. We next establish that this loss is inevitable for generic methods.  

In the statement below, the notation $A\circ B$ indicates the composition of algorithm $A$ after algorithm $B$. An $(\epsilon,\delta)$-approximation algorithm for a function $f$ with respect to some notion of distance is a randomized (or deterministic) algorithm that on every valid input $x$ returns with probability at least $1-\delta$ some value that is distance at most $\epsilon$ from $f(x)$. 
An algorithm $A$ is $(k, \delta)$-pseudodeterministic if for every $x$, there is a set $S_x$ of size at most $k$ and $A(x) \in S(x)$ with probability at least $1-\delta$. 

\begin{restatable}{theorem}{restatableNoFreeLunch}\label{no-free-lunch}
Let $d,k\in\N$ and $\epsilon_0\in(0,\infty)$ and $\delta\in(0,\frac12]$ be fixed, and let $\norm{\cdot}$ be a norm on $\R^d$. Suppose there exists $\epsilon\in(0,\infty)$ and a deterministic algorithm $A$ mapping inputs in $\R^d$ to outputs in $\R^d$ with the following universal black box property:
\begin{description}
    \item[Property:] For any set $\Lambda$ (indicating some problem domain) and function $f:\Lambda\to\R^d$ and $(\epsilon_0,\delta)$-approximation algorithm $B$ for $f$ (with respect to $\norm{\cdot}$), it holds that $A\circ B$ is a $(k,\delta)$-pseudodeterministic $(\epsilon,\delta)$-approximation algorithm for $f$ (again with respect to $\norm{\cdot}$).
\end{description}
Then $\epsilon\geq\epsilon_0\cdot\frac{d}{4\ln(2k)}$.
\end{restatable}

This theorem demonstrates that when $k(d)\in\poly(d)$, a blowup of at least $O\left(\frac{d}{\ln(d)}\right)$ is required of generic methods.

\subsection{Organization}
The rest of the document is organized as follows. \Autoref{sec:proof-outlines} contains proof outlines of our results. \Autoref{sec:notation} introduces the necessary notation. The rest of the sections are devoted to complete proofs. \Autoref{sec:general-bounds} contains proofs of degree $(k)$ lower bound results in terms of tolerance ($\epsilon$): \Autoref{epsilon-measure-bounded-any-norm}, \Autoref{epsilon-diameter-bounded-any-norm}, and \Autoref{k-exponential-in-dimension-measure}. \Autoref{sec:NewSperner} contains a proof of the neighborhood Sperner lemma (\Autoref{new-sperner-kkm-lebesgue}). \Autoref{sec:constructions} presents constructions of new secluded partitions (\Autoref{basic-reclusive-gluing}). Finally, \Autoref{sec:NoFreeLunch} proves the no-free-lunch result (\Autoref{no-free-lunch}).

%% file: 2_proof_outlines.tex
\section{Proof Outlines of Main Results}
\label{sec:proof-outlines}

The main technical tools that we use come from measure theory and the geometry of numbers and include the \namefullref{generalized-brunn-minkowski-inequality}, the \namefullref{isodiametric-ineq}, and the ideas from the proof of Blichfeldt's theorem (\Autoref{lower-bound-cover-number-Rd}). The generalized Brunn-Minkowski inequality gives a lower bound on the measure of a Minkowski sum of sets ($A+B=\set{\vec{a}+\vec{b}\colon \vec{a}\in A,\vec{b}\in B})$ based on the measures of those sets. The isodiametric inequality gives a very generic way to convert from upper bounds on the diameters of sets to upper bounds on measures in a tight manner. A common technique in the proof of Blichfeldt's theorem is to use an averaging argument to show that if a set $A$ is covered by a large family of other sets, then some point in $A$ is covered many times.

\subsection{Proof Outline for the Lower Bound on the Degree}

We discuss the main techniques to prove \Autoref{epsilon-measure-bounded-any-norm}. Since the theorem  holds for any norm, the proof uses no specific properties of any norm. However, to give the main ideas of the proof here and avoid technical details, we will use the visual of \Autoref{fig:cover-technique} and focus on the $\ell_\infty$ norm (whose balls are high-dimensional cubes) and assume that the bound on the measures is $M=1$. Since the volume of the unit ball in $\R^d$ with respect to $\ell_\infty$ is $v_d=2^d$,  \Autoref{epsilon-measure-bounded-any-norm} reduces to \Autoref{k-exponential-in-dimension-measure} mentioned in the introduction and restated below.

\restatableKExponentialInDimensionMeasure*

\begin{proof}[Proof Outline]
\renewcommand{\C}{\mathcal{C}}
\renewcommand{\A}{\mathcal{A}}
The goal is to find some point $\vec{p}\in\R^d$ such that $\balloinf{\varepsilon}{\vec{p}}$ intersects at least $(1+2\varepsilon)^d$ members of the partition. Instead of directly trying to establish this,  
we take a critical change of perspective: for any $\vec{p}\in \R^d$ and $X\in \P$ (or really any $X\subseteq\R^d$), it holds that $\balloinf{\varepsilon}{\vec{p}} \cap X \neq \emptyset$ if and only if $\vec{p}\in\bigcup_{\vec{x}\in X}\balloinf{\varepsilon}{\vec{x}}$\footnote{  
If $\balloinf{\varepsilon}{\vec{p}}$ intersects $X$, then there is some some point $\vec{x}$ in the intersection of these sets which means $\vec{p}$ and $\vec{x}$ are distance at most $\varepsilon$ apart, and we could also say $\vec{p}\in\balloinf{\varepsilon}{\vec{x}}$. Then trivially $\vec{p}\in\bigcup_{\vec{x}\in X}\balloinf{\varepsilon}{\vec{x}}$.
The converse is similar. If $\vec{p}\in\bigcup_{\vec{x}\in X}\balloinf{\varepsilon}{\vec{x}}$, then for some fixed $\vec{x}\in X$, $\vec{p}\in\balloinf{\varepsilon}{\vec{x}}$. Again, this means that $\vec{p}$ and $\vec{x}$ are distance at most $\varepsilon$ apart, so $\vec{x}\in\balloinf{\varepsilon}{\vec{p}}$. Since $\vec{x}\in X$ also, we have that $\balloinf{\varepsilon}{\vec{p}}$ and $X$ intersect.
}. Thus, what we do is to ``replace'' every member $X$ of the partition with the enlarged set $\bigcup_{\vec{x}\in X}\balloinf{\varepsilon}{\vec{x}}$ and try to find a point $\vec{p}$ that belongs to at least $(1+2\varepsilon)^d$ of these enlarged sets.
To achieve this, we take inspiration from a common proof of Blichfeldt's theorem---specifically, the following result which says that if we have a collection of sets $A_1,A_2,A_3,\ldots$ which are subsets of another set $S$, then there is a point in $S$ occurring in multiple $A_i$s provided together the $A_i$s have enough volume/measure. We can in fact give a lower bound on the number of sets $A_i$s to which such a point belongs to.
The following is the formal claim of this known result.

        \begin{restatable*}[Continuous Multi-Pigeonhole Principle]{proposition}{restatableLowerBoundCoverNumberRd}\label{lower-bound-cover-number-Rd}
        \renewcommand{\A}{\mathcal{A}}
        Let $d\in\N$ and $S\subset\R^d$ be bounded and measurable. Let $\A$ be a family of measurable subsets of $S$, and let $k=\ceil*{\frac{\sum_{A\in\A}m(A)}{m(S)}}$. Then if $k<\infty$, there exists $\vec{p}\in S$ such that $\vec{p}$ belongs to at least $k$ members of $\A$. (And if $k=\infty$, then for any $n\in\N$ there exists $\vec{p}^{(n)}\in S$ such that $\vec{p}^{(n)}$ belongs to at least $n$ members of $\A$.)
        \end{restatable*}

~ 

There is an immediate issue we have to deal with to be able to use \Autoref{lower-bound-cover-number-Rd} for our application. We would like to take $\A$ to be the collection of enlarged partition members: $\A=\set{\bigcup_{\vec{x}\in X}\balloinf{\varepsilon}{\vec{x}}}_{X\in\P}$, but then all we know is that each of these is a subset of $S=\R^d$ which is not bounded. This is a simple enough issue to deal with using a standard measure theory technique of considering instead a sequence $S_1,S_2,S_3,\ldots$ of sets which {\em are} bounded and get larger and larger so that $\bigcup_{n=1}^\infty S_n=\R^d$; we work with each of these sets individually and then try to use a limiting argument to pass the result back to $S=\R^d$. Specifically, we will take $S_n=[-n,n]^d$ as illustrated in the first 2 panes of \Autoref{fig:cover-technique}. The third pane of \Autoref{fig:cover-technique} illustrates that we will specifically consider the partition of $S_n$ induced by $\P$ which we denote by $\S_n$. That is, the induced partition $\S_n$  is the set $\set{X\cap S_n\colon X\in\P\text{ and }X\cap S_n\not=\emptyset}$. Then for each $S_n$ we consider a collection $\A_n$ of the enlarged members of the induced partition: $\A_n=\set{A_Y}_{Y\in\S_n}$ where $A_Y\defeq \bigcup_{\vec{y}\in Y}\balloinf{\varepsilon}{\vec{y}}$. Note that each $A_Y$ is a subset of $S_n'\defeq [-(n+\varepsilon),n+\varepsilon]^d$.

However, there remains one other issue to deal with to utilize \Autoref{lower-bound-cover-number-Rd}---for each $n$, we have to have some lower bound on the expression $\ceil*{\frac{\sum_{A_Y\in\A_n}m(A_Y)}{m(S_n')}}$. We know that $m(S_n')=(2(n+\varepsilon))^d$, and using the fact that $\S_n$ is a partition of $S_n=[-n,n]^d$, we have the following if $\S_n$ is countable\footnote{
If $\S_n$ is uncountable, then the first equality below becomes an inequality, but it becomes the wrong inequality: $\sum_{Y\in\S_n}m(Y) \leq m\left(\bigsqcup_{Y\in\S_n}Y\right)$ by \Autoref{disjoint-uncountable} in \Autoref{sec:measure-theory}.
} (meaning finite or countably infinite):

\[
\sum_{A_Y\in\A_n}m(A_Y) \geq \sum_{Y\in\S_n}m(Y) = m\left(\bigsqcup_{Y\in\S_n}Y\right) = m(S_n),
\]
but this is not nearly good enough, because it just gives 
\[
\ceil*{\frac{\sum_{A_Y\in\A_n}m(A)}{m(S_n')}}
\geq \ceil*{\frac{m(S_n)}{m(S_n')}}
\geq \ceil*{\frac{(2n)^d}{(2(n+\varepsilon))^d}} 
= \ceil*{\left(\frac{n}{n+\varepsilon}\right)^d} 
= 1
\]
whereas we want it $\geq(1+2\varepsilon)^d$. Basically, this lower bound is terrible because we did not account for the fact that the elements of $\A_n$ are enlarged from what they were in the partition $\S_n$. Thus, we really want some way to give for each $Y\in\S_n$, a lower bound on the measure of the enlarged set $A_Y$. One might observe that enlarging with an $\varepsilon$-ball looks something like scaling by a factor of $1+\varepsilon$ (though it is not actually scaling\footnote{
See for example the smallest (purple) member in the last pane of \Autoref{fig:cover-technique} which is very circular, but upon enlarging looks more squarish. In fact, enlarging a single point by $\varepsilon$ results in a hypercube of diameter $2\varepsilon$.
}), and since the Lebesgue measure (i.e. typical notion of volume/measure in $\R^d$) has the property that scaling by $(1+\varepsilon)$ increases the measure by a factor of $(1+\varepsilon)^d$, we might be able to show that the enlarged version of each member increases by a factor of $(1+\varepsilon)^d$ (which is basically what we are looking to get).

This intuition holds, though the actual reason is not related to scaling, and is dependent on the members having measure at most $1$. Rather, we use a specialized adaption of the known \namefullref{generalized-brunn-minkowski-inequality} to show that
\begin{equation}\label{eq:brunn-minkowski-introduction}
m(A_Y) \geq m(Y) \cdot (1 + 2\varepsilon)^d
\end{equation}
holds\footnote{Technically speaking, the set $A_Y$ might not be measurable (though we suspect it is), so the expression may not even be a mathematically valid one to write down, but we deal with this technical detail later in the paper, and it really does not effect anything in this proof outline.}.
Now that we have dealt with both issues that arise with trying to apply \Autoref{lower-bound-cover-number-Rd}, we can consider a fixed $n$ and can continue. We proceed in two cases: (1) the interesting case in which $\S_n$ has only countably many members, and (2) the nearly trivial case in which the partition $\S_n$ contains uncountable many members. In case (1)  we have
\begin{align*}
    \ceil*{\frac{\sum_{A_Y\in\A_n}m(A_Y)}{m(S_n')}} &=    \ceil*{\frac{\sum_{Y\in\S_n}m(A_Y)}{m(S_n')}} \tag{Re-index}\\
    &\geq \ceil*{\frac{\sum_{Y\in\S_n}\left[m(Y) \cdot (1 + 2\varepsilon)^d\right]}{m(S_n')}} \tag{\Autoref{eq:brunn-minkowski-introduction}}\\
    &=\ceil*{\frac{(1 + 2\varepsilon)^d\cdot\sum_{Y\in\S_n}m(Y)}{m(S_n')}} \tag{Linearity of summation}\\
    &=\ceil*{\frac{(1 + 2\varepsilon)^d\cdot m\left(\bigsqcup_{Y\in\S_n}Y\right)}{m(S_n')}} \tag{Countable additivity of measures}\\
    &=\ceil*{\frac{(1 + 2\varepsilon)^d\cdot m(S_n)}{m(S_n')}} \tag{$S_n=\bigsqcup_{Y\in\S_n}Y$}\\
    &=\ceil*{(1 + 2\varepsilon)^d\cdot \left(\frac{n}{n+\varepsilon}\right)^d} \tag{$\frac{m(S_n)}{m(S_n')}=\frac{n}{n+\varepsilon}$ as above} \\
\end{align*}
Since $\lim_{n\to\infty}\left(\frac{n}{n+\varepsilon}\right)^d=1$, then $\lim_{n\to\infty}(1 + 2\varepsilon)^d\cdot \left(\frac{n}{n+\varepsilon}\right)^d=(1 + 2\varepsilon)^d$, so because there is a ceiling involved, we can take $N\in\N$ to be large enough that $\ceil*{(1 + 2\varepsilon)^d\cdot \left(\frac{N}{N+\varepsilon}\right)^d} = \ceil*{(1 + 2\varepsilon)^d}$ (see \hyperlink{hypertarget-smaller-ceiling}{\Autoref*{smaller-ceiling}} in \hyperlink{hypertarget-smaller-ceiling}{\Autoref*{sec:other}}), so by \Autoref{lower-bound-cover-number-Rd}, there is a point $\vec{p}\in S_n'$ that is contained in at least $(1 + 2\varepsilon)^d$ many sets in $\A_N$, and by our change of perspective, this point $\vec{p}$ has the property that $\balloinf{\varepsilon}{\vec{p}}$ intersects at least $(1 + 2\varepsilon)^d$ many members of $\P$.

In case (2) where some $\S_N$ contains uncountably many members, then we completely ignore the lower bound for $m(A_Y)$ in \Autoref{eq:brunn-minkowski-introduction} because it might be that lots of members $Y$ (possibly all of them) have measure $0$, and so that bound only tells us that $m(A_Y)\geq 0$. Instead, we note that $Y$ is at least non-empty, so contains at least one point $\vec{y}$, and thus $A_Y\supseteq \balloinf{\varepsilon}{\vec{y}} = \prod_{i=1}^{d}[y_i-\varepsilon, y_i+\varepsilon]$, and so $m(A_Y)\geq (2\varepsilon)^d$. Thus, $\ceil*{\frac{\sum_{A_Y\in\A_N}m(A_Y)}{m(S_N')}}=\infty$, so by \Autoref{lower-bound-cover-number-Rd}, there is a point $\vec{p}\in S_N'$ that is contained in at least $(1 + 2\varepsilon)^d$ many sets in $\A_N$, and by our change of perspective, this point $\vec{p}$ has the property that $\balloinf{\varepsilon}{\vec{p}}$ intersects at least $(1 + 2\varepsilon)^d$ many members of $\P$.
\end{proof}

We now outline how \Autoref{epsilon-diameter-bounded-any-norm} follows from the theorem. As mentioned in the introduction, for the $\ell_\infty$ norm, a bound of $D$ on the diameter of a set $X\subseteq\R^d$ implies a bound of $M=D^d$ on the outer measure of $X$ because a special property of the $\ell_\infty$ norm is that $X$ is actually contained within some shift of $[0,D]^d$ which has measure $M=D^d$ (see \Autoref{diameter-ball}), so in particular, diameter at most $1$ implies outer measure at most $1$. This containment fact does not hold in general for other norms.

Nonetheless, for any norm, a diameter bound of $D$ implies {\em some} measure bound. In particular, we can place a ball of {\em radius} $D$ at any point in the set and know that it is contained in the ball showing that the outer measure is at most $m(\ballcstd{D}{\vec{0}})$. However, we can get a better bound using the known \namefullref{isodiametric-ineq} which says that we can actually bound the measure by $m(\ballcstd{D/2}{\vec{0}})$. In other words, while the containment property of the $\ell_\infty$ norm does not hold for general norms, it at least holds {\em in spirit}; a set of diameter $D$ might not be contained in any ball of diameter $D$, but it can be cut up somehow to fit in the ball of diameter $D$ (radius $D/2$), so it has no greater measure. With this inequality in hand and the specific values of the diameter that it gives, \Autoref{epsilon-diameter-bounded-any-norm} follows from \Autoref{epsilon-measure-bounded-any-norm}.

\subsection{Proof Outline of the Neighborhood Sperner/KKM/Lebesgue}
The main idea behind the proof of \Autoref{new-sperner-kkm-lebesgue} is the same as discussed in the previous subsection. For each color $C$, let $X_C$ be the set of points that are colored $C$. We union an $\varepsilon$-ball at each point in $X_C$, to obtain an enlarged version $X_C$. Now, as before, by the \namefullref{lower-bound-cover-number-Rd} there is a point that belongs to many of the enlarged sets.  

However, there are some additional issues that arise on the unit cube that don't arise in $\R^d$. In the discussion above, the enlarged set was not contained in the original region (denoted $S_n'$ above) and we needed to consider a larger region (denote $S_n$ above) to contain them. In $\R^d$ we could deal with this via a limiting argument so that the ratio of the volume change $m(S_n')/m(S_n)$ tends to $1$ (i.e. it became negligible when ceilings were involved). If one enlarges every color in a unit cube $[-\frac12,\frac12]^d$ in the same way, the measure of each color is guaranteed to increase by a factor of $(1+2\epsilon)^d$ as before, but also the smallest set that contains all of these sums is the unit cube $[-\frac12-\epsilon,\frac12+\epsilon]^d$ which increased in measure by a factor of $(1+2\epsilon)^d$ compared to the original cube, so nothing has been gained! Obviously there will be an overlap of the sums, but the bounds given by the generalized Brunn-Minkowsi inequality tell us no additional information.

We deal with this by employing a trick of first extending the coloring directly to $[-\frac12-\epsilon,\frac12+\epsilon]^d$ in a natural way that ensures each color is bounded away from the boundary so that we can perform an enlargement using just one orthant of the $\epsilon$-ball (instead of the whole ball) and still have the enlarged color set contained in $[-\frac12-\epsilon,\frac12+\epsilon]^d$. This means we end up knowing that each color has increased in measure by at least a factor of $(1+\frac23\epsilon)^d$ and that the containing region has not changed in measure at all.

\subsection{Proof Outline of the Construction Result}

The approach to the construction of new secluded partitions of $\R^d$ is the following: (1) break up the coordinates and view $\R^d$ as $\R^{d_1}\times\R^{d_2}\times\cdots\times\R^{d_n}$ for some $d_1,d_2,\ldots,d_n$ with $d_1+d_2+\cdots+d_n=d$, (2) partition each $\R^{d_i}$ independently into some $\P_i$, and finally (3) combine the partitions into one partition $\P$ of $\R^d$ in a natural way by taking $\P$ to be $\set{\prod_{i=1}^{n}X_i\colon X_i\in\P_i}$. This is all described in \Autoref{defn:partition-product}. This construction can be done with arbitrary partitions, but it is straightforward to argue that if each $\P_i$ is $(k_i,\varepsilon_i)$-secluded, then $\P$ is $(k,\varepsilon)$-secluded for $k=\prod_{i=1}^n k_i$ and $\varepsilon=\min_{i\in[n]}\varepsilon_i$. This is shown in \Autoref{secluded-partition-product-proposition}.

We then simplify to the case where all $d_i$ are equal to $d'$ or $d'-1$ for some natural number $d'$, and we analyze the construction where each $\P_i$ is a $(d_i+1,\frac{1}{2d_i})$-secluded unit cube partition (as in \cite{geometry_of_rounding}) with $d_i = d'$ or $d'-1$. The resultant partition of $\R^d$ is $(k,\varepsilon)$-secluded for $k=\prod_{i=1}^n (d_i+1)\leq (d'+1)^{\ceil*{\frac{d}{d'}}}$ and $\varepsilon=\min_{i\in[n]}\frac1{2d_i}=\frac{1}{2d'}$. We use  $f:\N\to\N$ to prescribe $d'$ as a function of $d$ to obtain \Autoref{basic-reclusive-gluing}. 

A final note is that one does have to be careful to keep the ceiling function in this result because it keeps the exponent from becoming less than $1$.

\subsection{Proof Outline of the No-Free-Lunch Theorem}
We use \Autoref{epsilon-diameter-bounded-any-norm} to prove \Autoref{no-free-lunch}. Since $A$ is a deterministic algorithm, it can be viewed as a function $A:\R^d\to\R^d$. Such a  function induces a partition of its domain by considering points to be equivalent if they are mapped to the same value. We can show that because $A$ gives $\epsilon$-approximations (with respect to a norm $\norm{\cdot}$), each member of the partition has diameter at most $2\epsilon$ (with respect to $\norm{\cdot}$). By \Autoref{epsilon-diameter-bounded-any-norm}, this means that there is some point $\vec{p}$ in the domain $\R^d$ such that 
$\ballostd{\frac{\varepsilon_0}{2}}{\vec{p}}$  intersects at least $\left(1+\frac{2(\epsilon_0/2)}{2\epsilon}\right)^d=\left(1+\frac{\epsilon_0}{2\epsilon}\right)^d$ members of the partition. In other words, there is a point $\vec{p}$ and a set of at least $\left(1+\frac{\epsilon_0}{2\epsilon}\right)^d$ points in $\ballostd{\frac{\varepsilon_0}{2}}{\vec{p}}$ such that $A$ maps each of these points to a different value. 

Then we consider one specific problem and approximation algorithm and prove that $A$ can't perform very well on it. Let  $f:\Lambda\to\R^d$ be such that  there is  $\lambda\in\Lambda$ with $f(\lambda)=\vec{p}$. Then we consider an $(\epsilon_0,\delta)$-approximation algorithm $B$ for this function $f$ which has the specific property that the output distribution of $B(\lambda)$ is uniform over the $\left(1+\frac{\epsilon_0}{2\epsilon}\right)^d$ points above (which all $\epsilon_0$-approximate $f(\lambda)=\vec{p}$). We show that in order for $A\circ B$ to be $(k,\delta)$-pseudodeterministic requires very large $k$ because $B$ is distributing uniformly over a very large set. Specifically, $k\geq(1-\delta)\cdot \left(1+\frac{\epsilon_0}{2\epsilon}\right)^d$. 
Finally, we rearrange this expression to solve for $\epsilon$ and use the approximation $\ln(1+x)\geq\frac{x}2$ for small $x$, and the fact that $(1-\delta)\geq\frac12$ to arrive at the stated lower bound on $\epsilon$.

\begin{figure}
    \centering
    \includegraphics[width=\textwidth]{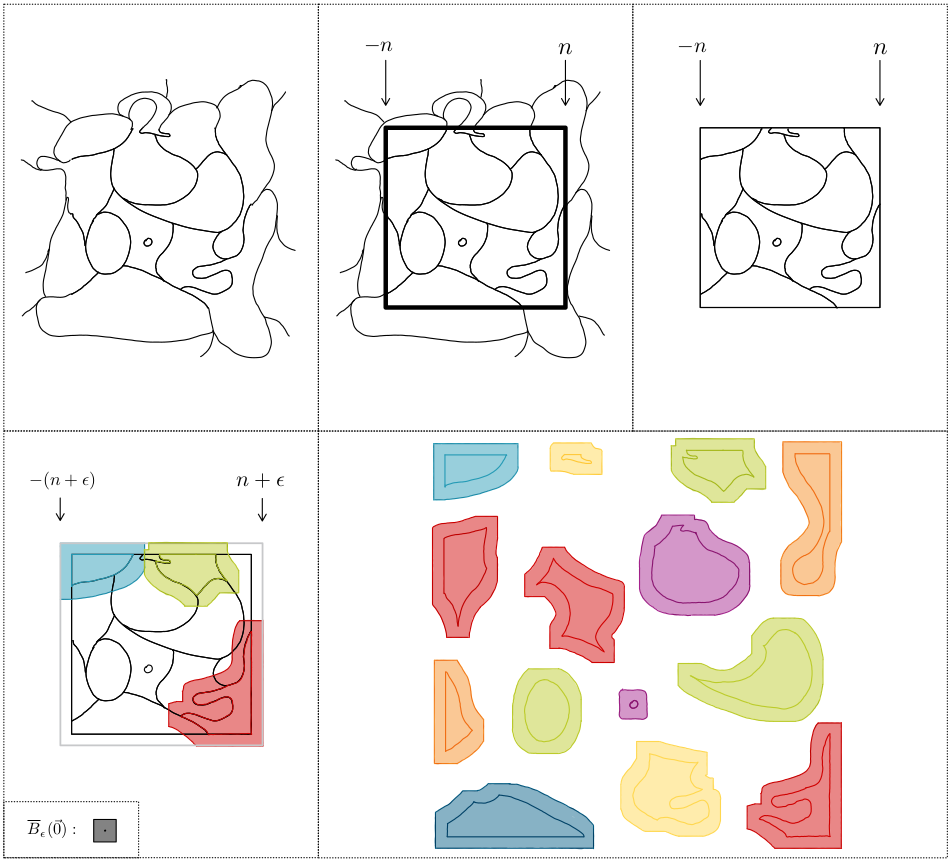}
    \caption{In the first pane, we have a partition of $\R^2$. In the second panes, we show that we will consider only members of the partition intersect $[-n,n]^2$, and in the third pane we show the partition that $\P$ induces on $[-n,n]^2$. In the fourth pane, we consider enlarging each member by placing at $\varepsilon$-ball at each point of the member and show that these enlarged elements are still contained within $[-(n+\varepsilon),n+\varepsilon]^2$. In the fifth pane, we see all of the expanded members and observe that the sum of the areas of the enlarged members is ``significantly'' more that the area of $[-n,n]^2$.}
    \label{fig:cover-technique}
\end{figure}

%% file: 3_notation.tex
\newpage
\section{Notation}\label{sec:notation}

The following is a list of some of the notation we will use in this paper.

\begin{itemize}
    \item We use $\N$ to denote the natural numbers starting with $1$.
    \item We continue to use $\ballcstd{\epsilon}{\vec{p}}$, $\ballostd{\epsilon}{\vec{p}}$, $\ballcinf{\epsilon}{\vec{p}}$, and $\balloinf{\epsilon}{\vec{p}}$ as before.
    \item For two sets $A,B\subseteq\R^d$ we write $A+B$ to represent the Minkowski sum $A+B\defeq\set{\vec{a}+\vec{b}\colon \vec{a}\in A,\;\vec{b}\in B}$. We also may write $\vec{a}+B$ to mean $\set{\vec{a}+\vec{b}\colon \vec{b}\in B}$ for some fixed vector $\vec{a}$.
      \item We will use $\unitballmeasurestd$ to represent the Borel/Lebesgue measure of the unit radius ball in $\R^d$ with respect to a general norm $\norm{\cdot}$. This is a normalization factor that appears in some results.
 
\end{itemize}

%% file: 4_general_bounds.tex
\section{Lower Bound on the Degree Parameter}
\label{sec:general-bounds}

In this section, we present a complete proof of \Autoref{epsilon-measure-bounded-any-norm}. We begin with some prerequisite results in \Autoref{subsec:prerequisite-results}. Then, in \Autoref{epsilon-measure-bounded-any-norm} we present the proof of \Autoref{sec:Main-Results}. We follow this immediately by proving the two consequences mentioned in the introduction: \Autoref{epsilon-diameter-bounded-any-norm} and 
\Autoref{k-exponential-in-dimension-measure}.

\subsection{Prerequisite Results}
\label{subsec:prerequisite-results}

In this section, we will deal with arbitrary norms of $\R^d$. We  point out the well-known fact that all norms on $\R^d$ are equivalent in the sense that they all generate the same topology on $\R^d$. Given two norms $\norm{\cdot}^a$ and $\norm{\cdot}^b$ on $\R^d$, there exists fixed constants $c_d,C_d\in(0,\infty)$ such that for all vectors $\vec{x}\in\R^d$, it holds that $c_d\norm{\vec{x}}^a \leq \norm{\vec{x}}^b \leq C_d\norm{\vec{x}}^a$. Thus the collection of open sets in $\R^d$ is the same no matter which norm we are using. This also means that the Borel and Lebesgue $\sigma$-algebras on $\R^d$ are the same no matter which norm is used, and thus balls with respect to any norm on $\R^d$ are measurable.

We begin with four simple facts. All  have straightforward proofs which we provide in \hyperlink{hypertarget-smaller-ceiling}{\Autoref*{sec:other}} and  \Autoref{sec:minkowski-sums}. The first fact will later allow us to pass a result through a limit since the answer will be an integer.

\begin{restatable}{fact}{restatableSmallerCeiling}\label{smaller-ceiling}
For any $\alpha\in\R$, there exists $\gamma\in\R$ such that $\gamma<\alpha$ and $\ceil{\gamma}=\ceil{\alpha}$.
\end{restatable}

The next fact says that the Minkowski sum of a set $X$ and an open ball at the origin can be viewed as a union of open balls positioned at each point of $X$.

\begin{restatable}{fact}{restatableMinkowskiBubbleUnion}\label{minkowski-bubble-union}
For any normed vector space, given a set $X$ and $\varepsilon\in[0,\infty)$, then
\[
X+\ballostd{\varepsilon}{\vec{0}} = \bigcup_{\vec{x}\in X}\ballostd{\varepsilon}{\vec{x}}.
\]
The same can be said replacing open balls with closed balls.
\end{restatable}

The third fact says that we can decompose a ball into a Minkowski sum of two smaller balls.

\begin{restatable}{fact}{restatableSumOfBalls}\label{sum-of-balls}
For any normed vector space, and any $\alpha,\beta\in(0,\infty)$, it holds that $\ballostd{\alpha}{\vec{0}}+\ballostd{\beta}{\vec{0}}=\ballostd{\alpha+\beta}{\vec{0}}$.
\end{restatable}

The fourth fact, while also very simple, is the key change of perspective that allowed us to prove the main results of this section. It says that if we are checking which sets $X$ in our partition intersect an $\varepsilon$-ball located at $\vec{p}$ (in order to see how many there are), we can instead enlarge each member of the partition by taking its Minkowski sum with the origin-centered $\varepsilon$-ball, and check which of these enlarged members contain the point $\vec{p}$.

\begin{restatable}{fact}{restatableAntisymmetryContainmentPreClaimMinkowskiSum}\label{antisymmetry-containment-pre-claim-minkowski-sum}
For any normed vector space, for any set $X$, for any vector $\vec{p}$, and any $\epsilon>0$, the following are equivalent:
\begin{enumerate}
    \item $\ballostd{\varepsilon}{\vec{p}}\cap X\not=\emptyset$
    \item $\vec{p}\in X+\ballostd{\varepsilon}{\vec{0}}$
\end{enumerate}
The same can be said replacing both open balls with closed balls.
\end{restatable}

Now we introduce the result which is the connection to the above-mentioned key change of perspective. The result says to consider a bounded, (measurable) subset $S\subseteq\R^d$ (so it has finite measure) and a collection $\mathcal{A}$ of (measurable) subsets of $S$. If we compute the sum of measures of all members in the collection $\mathcal{A}$ (i.e. intuitively the total volume that they take up), and compare this to the measure/volume of $S$, then whatever this ratio is, we can find a point in $S$ covered by that many members of the collection $\mathcal{A}$. 
For example, in the simplest case that the total measure of members of $\mathcal{A}$ is larger than the measure of $S$, then there is no way for all of the members of $\mathcal{A}$ to be disjoint, so there has to be some point covered by two members. This simple case can be viewed as a continuous version of the pigeonhole principle.

In the more generic case, this result should be intuitively true by an averaging argument: if every point of $S$ is covered only $n$ times, then the total measure of members in $\mathcal{A}$ is at most $n\cdot m(S)$, so if the ratio of total measure in $\mathcal{A}$ to the measure of $S$ is large, then $n$ must also be large. This more general version is a sort of continuous multi-pigeonhole principle.

\restatableLowerBoundCoverNumberRd*

While this result may be intuitive, proving it formally does require some effort. We first encountered this result as the main ingredient in the standard proof of Blichfeldt's theorem (which was the source of motivation for our main technique), but many of the sources we found where proofs of Blichfeldt's theorem are presented did not prove the result above except in special cases, so for convenience and completion, we provide a proof in \Autoref{sec:measure-theory} in three parts: \Autoref{exact-measure-of-multiplicity}, \Autoref{upper-bound-measure-of-multiplicity}, and \Autoref{lower-bound-cover-number}.

The next ingredient that we need is a way to measure how large the Minkowski sum in \Autoref{antisymmetry-containment-pre-claim-minkowski-sum} is. In order to utilize \Autoref{lower-bound-cover-number-Rd} we need a lower bound on the measures, and we can obtain one using the generalized Brunn-Minkowski inequality stated below.

\begin{restatable}[Generalized Brunn-Minkowski Inequality]{theorem}{restatable-generalized-brunn-minkowski-inequality}\label{generalized-brunn-minkowski-inequality}
Let $d\in\N$ and $A,B\subseteq\R^d$ be Lebesgue measurable such that $A+B$ is also Lebesgue measurable. Then
\[
m(A+B) \geq \left[m(A)^\frac1d + m(B)^\frac1d\right]^d.
\]
\end{restatable}

This version of the statement can be obtained from \cite[Equation~11]{gardner_brunn-minkowski_2002}; in that survey, Gardner states this theorem with a requirement that the sets be bounded, but in the following paragraph notes that this is not necessary and the requirement is only stated for convenience of the presentation in that survey.

In the theorem, the requirement that $A+B$ is Lebesgue measurable is not a triviality; Gardner discusses that there exist known Lebesgue measurable sets $A$ and $B$ such that the Minkowski sum $A+B$ is not Lebesgue measurable as shown in \cite{sierpinski_sur_1920}. The next result gives us a way to circumvent this issue in our application even if the members of our partition are not measurable by taking $B$ to be an open set so that the sum $A+B$ is open (and thus measurable), and using the outer measure of $A$ so that we don't need the assumption that $A$ is measurable.

\begin{lemma}\label{outer-measure-brunn-minkowsi-bound}
Let $d\in\N$ and let $\R^d$ be equipped with any norm $\norm{\cdot}$. Let $Y\subseteq\R^d$, and $\varepsilon\in(0,\infty)$. Then $Y+\ballostd{\varepsilon}{\vec{0}}$ is open (and thus Borel measurable), and $m(Y+\ballo{\varepsilon}{\vec{0}})\geq \left(m_{out}(Y)^\frac1d+\varepsilon\cdot (\unitballmeasurestd)^\frac1d\right)^d$.
\end{lemma}
\begin{proof}
By \Autoref{minkowski-bubble-union}, for any $\epsilon''\in(0,\infty)$, $Y+\ballostd{\varepsilon''}{\vec{0}}=\bigcup_{\vec{y}\in Y}\ballostd{\varepsilon''}{\vec{y}}$  which is a union of open sets, so is itself open and thus Borel measurable. Now, for any $\varepsilon'\in(0,\varepsilon)$, observe that by \Autoref{sum-of-balls}, $\ballostd{\varepsilon}{\vec{0}}=\ballostd{\varepsilon-\varepsilon'}{\vec{0}}+\ballostd{\varepsilon'}{\vec{0}}$ and thus, this sum is measurable because it is an open ball. Using this equality and the associativity of the Minkowski sum, we have 
\[
Y+\ballostd{\varepsilon}{\vec{0}} = Y + \left[\ballostd{\varepsilon-\varepsilon'}{\vec{0}}+\ballostd{\varepsilon'}{\vec{0}}\right] = \left[Y + \ballostd{\varepsilon-\varepsilon'}{\vec{0}}\right] +\ballostd{\varepsilon'}{\vec{0}}.
\]
Thus, we have the following chain of inequalities (each justified after it is stated):
\begin{align*}
    m\left(Y+\ballostd{\varepsilon}{\vec{0}}\right) &= m\left(\left[Y + \ballostd{\varepsilon-\varepsilon'}{\vec{0}}\right] +\ballostd{\varepsilon'}{\vec{0}}\right) \tag{Open, measurable, equality above} \\
    &\geq 
    \left( 
        m\left(Y + \ballostd{\varepsilon-\varepsilon'}{\vec{0}}\right)^{\frac1d}
        +
        m\left(\ballostd{\varepsilon'}{\vec{0}}\right)^{\frac1d}
    \right)^d
    \intertext{The above comes from the \namefullref{generalized-brunn-minkowski-inequality} noting that as demonstrated above, terms of the sum $\left[Y + \ballostd{\varepsilon-\varepsilon'}{\vec{0}}\right]$ and $\ballostd{\varepsilon'}{\vec{0}}$ are open and thus measurable, and the same holds for the sum itself $\left(Y+\ballostd{\varepsilon}{\vec{0}}\right)$. We continue.}
    &\geq
    \left( 
        m_{out}\left(Y\right)^{\frac1d}
        +
        m\left(\ballostd{\varepsilon'}{\vec{0}}\right)^{\frac1d}
    \right)^d
    \intertext{The above inequality comes from the definition of the outer measure of $Y$ being the infimum of the measures of all measurable supersets of $Y$. Since $Y\subseteq Y+\ballostd{\varepsilon'}{\vec{0}}$, we get the inequality above. Continuing, we have the following:}
    &=
    \left( 
        m_{out}\left(Y\right)^{\frac1d}
        +
        m\left(\varepsilon'\cdot\ballostd{1}{\vec{0}}\right)^{\frac1d}
    \right)^d
    \tag{Scaling of norm-based balls} \\
    &=
    \left( 
        m_{out}\left(Y\right)^{\frac1d}
        +
        \left[(\varepsilon')^d \cdot m\left(\ballostd{1}{\vec{0}}\right)\right]^{\frac1d}
    \right)^d
    \tag{Scaling for Lebesgue measure} \\
    &=
    \left( 
        m_{out}\left(Y\right)^{\frac1d}
        +
        \varepsilon' \cdot (\unitballmeasurestd)^{\frac1d}
    \right)^d
    \tag{Algebra and $\unitballmeasurestd\defeq m\left(\ballostd{1}{\vec{0}}\right)$}
\end{align*}
Since the inequality above holds for all $\varepsilon'\in(0,\varepsilon)$, it must also hold in the limit (keeping $d$ and $Y$ fixed):
\[
m\left(Y+\ballostd{\varepsilon}{\vec{0}}\right) \geq \lim_{\varepsilon'\to\varepsilon}\left[\left(m_{out}\left(Y\right)^{\frac1d}+\varepsilon' \cdot (\unitballmeasurestd)^{\frac1d}\right)^d\right] = \left(m_{out}\left(Y\right)^{\frac1d}+\varepsilon \cdot (\unitballmeasurestd)^{\frac1d}\right)^d
\]
which concludes the proof.
\end{proof}

At this point we are nearly in a position to prove the main result of this section, but we need one last result which gives an inequality that we will compose with the \namefullref{generalized-brunn-minkowski-inequality}. The result below can be interpreted as saying that for appropriate parameters, we can essentially factor our the ``$x$'' in $(x^{1/d}+\alpha)^d$ to get the no larger expression $x(1+\alpha)^d$.

\begin{restatable}{lemma}{restatableBrunnMinkowskiBoundLemma}\label{brunn-minkowski-bound-lemma}
For $d\in[1,\infty)$, $x\in[0,1]$, and $\alpha\in[0,\infty)$, it holds that $(x^{1/d}+\alpha)^d\geq x(1+\alpha)^d$.
\end{restatable}

A proof of the inequality using standard calculus techniques can be found in \hyperlink{hypertarget-brunn-minkowski-bound-lemma}{\Autoref*{sec:other}}.

\subsection{Proofs of  Theorem~\ref*{epsilon-measure-bounded-any-norm}, Corollary~\ref*{epsilon-diameter-bounded-any-norm}, and Theorem~\ref*{k-exponential-in-dimension-measure}}\label{sec:Main-Results}
Now we restate and prove the main result.

\restatableMeasureBoundedAnyNorm*
\begin{proof}
\renewcommand{\A}{\mathcal{A}}
\newcommand{\F}{\mathcal{F}}

Throughout the proof, all lengths will be with respect to $\norm{\cdot}$, so we will eliminate the clutter by neglecting to use the $\norm{\cdot}$ subscript anywhere in the proof. We will also be working in a single dimension $d$ throughout the proof, so we also drop the $d$ subscript from $v$ throughout.

Consider the following for each $n\in\N$. Let $S_n=\ballo{n}{\vec{0}}$ and $S_n'=\ballo{n+\varepsilon}{\vec{0}}=S_n+\ballo{\varepsilon}{\vec{0}}$ and $\S$ be the partition of $S_n$ induced\footnote{
I.e. $\S=\set{X\cap S_n\colon X\in\P\text{ and }X\cap S_n\not=\emptyset}$. That is, $\S$ is the partition of $S_n$ obtained by intersecting every member of $\P$ with the new domain $S_n$ and keeping those that have non-empty intersections.
} by $\P$. For each $Y\in\S_n$, let $C_{Y}=Y+\ballo{\varepsilon}{\vec{0}}$. By \Autoref{outer-measure-brunn-minkowsi-bound}, $C_Y$ is measurable, and $m(C_Y)\geq \left(m_{out}\left(Y\right)^{\frac1d} + \varepsilon\cdot v^{\frac1d}\right)^d$. Also observe that $C_Y\subseteq S_n'$. Now consider the following inequalities:
\begin{align*}
    m(C_Y) &\geq \left(m_{out}\left(Y\right)^{\frac1d} + \varepsilon\cdot v^{\frac1d}\right)^d \tag{Above} \\  %
    &= \left(M^{1/d}\left[\frac{m_{out}\left(Y\right)^{\frac1d}}{M^{1/d}} + \frac{\varepsilon\cdot v^{\frac1d}}{M^{1/d}}\right]\right)^d \tag{Algebra} \\
    &= M\left(\left[\frac{m_{out}\left(Y\right)}{M}\right]^{\frac1d} + \frac{\varepsilon\cdot v^{\frac1d}}{M^{1/d}}\right)^d \tag{Algebra} \\
    &\geq M\cdot \frac{m_{out}\left(Y\right)}{M} \cdot \left(1+\frac{\varepsilon\cdot v^{\frac1d}}{M^{1/d}}\right)^d \tag{\Autoref{brunn-minkowski-bound-lemma} since $\frac{m_{out}\left(Y\right)}{M}\in[0,1]$} \\
    &= m_{out}\left(Y\right)\cdot\left(1+\frac{\varepsilon\cdot v^{\frac1d}}{M^{1/d}}\right)^d \tag{Simplify} \\
\end{align*}

Informally, the above shows that for each $Y\in\S_n$, the set $Y+\ballo{\varepsilon}{\vec{0}}$ has substantially more (outer) measure than $Y$ does---specifically a factor of $\left(1+\frac{\varepsilon\cdot v^{\frac1d}}{M^{1/d}}\right)^d$. We will extend this to unsurprisingly show that this implies that $\set{Y+\ballo{\varepsilon}{\vec{0}}}_{Y\in\S_n}$ also has this same factor more (outer) measure than $\S_n$ does, observing that $\S_n$ has total (outer) measure of about $m(S_n)$ since $\S_n$ is a partition of $S_n$ (any discrepancy is due to non-measurable members in $\S_n$)

Formally, we claim that there exists a finite subfamily $\F_n\subseteq\S_n$ such that
\[
\sum_{Y\in\F_n}m\left(Y+\ballo{\varepsilon}{\vec{0}}\right) \geq \left(1+\frac{\varepsilon\cdot v^{\frac1d}}{M^{1/d}}\right)^d\cdot m(S_n).
\]
To see this, first consider the case that $\S_n$ has infinite cardinality. Let $\F_n\subset\S_n$ be any subfamily of finite cardinality at least $\left(1+\frac{\varepsilon\cdot v^{\frac1d}}{M^{1/d}}\right)^d\cdot m(S_n)\cdot \frac{1}{\varepsilon^d v}$. This gives
\begin{align*}
    \sum_{Y\in\F_n}m\left(Y+\ballo{\varepsilon}{\vec{0}}\right) &\geq \sum_{Y\in\F_n}m\left(\ballo{\varepsilon}{\vec{0}}\right) \intertext{where this inequality is because $Y+\ballo{\varepsilon}{\vec{0}}$ is a superset of some translation of $\ballo{\varepsilon}{\vec{0}}$ since $Y\not=\emptyset$. Continuing, we use the standard fact that $m\left(\ballo{\varepsilon}{\vec{0}}\right)=m\left(\varepsilon\cdot\ballo{1}{\vec{0}}\right)=\varepsilon^d\cdot m\left(\ballo{1}{\vec{0}}\right)=\varepsilon^d v$:}
    &\geq \sum_{Y\in\F_n}\varepsilon^d v \\
    &= \left[\left(1+\frac{\varepsilon\cdot v^{\frac1d}}{M^{1/d}}\right)^d\cdot m(S_n)\cdot \frac{1}{\varepsilon^d v}\right] \cdot \varepsilon^d v \tag{Cardinality of $\F_n$}\\
    &= \left(1+\frac{\varepsilon\cdot v^{\frac1d}}{M^{1/d}}\right)^d\cdot m(S_n) \tag{Simplify}.
\end{align*}
Now consider the other (more interesting) case where $\S_n$ has finite cardinality\footnote{In fact this case also works if $\S_n$ is countable even though we have already dealt with that case.}. Take $\F_n=\S_n$ so that
\begin{align*}
    \sum_{Y\in\F_n}m\left(Y+\ballo{\varepsilon}{\vec{0}}\right) &=     \sum_{Y\in\S_n}m\left(Y+\ballo{\varepsilon}{\vec{0}}\right) \tag{$\F_n=\S_n$} \\
    &\geq \sum_{Y\in\S_n}m_{out}\left(Y\right)\cdot\left(1+\frac{\varepsilon\cdot v^{\frac1d}}{M^{1/d}}\right)^d \tag{Above}\\
    &= \left(1+\frac{\varepsilon\cdot v^{\frac1d}}{M^{1/d}}\right)^d \cdot \sum_{Y\in\S_n}m_{out}\left(Y\right) \tag{Linearity of summation} \\
    &\geq \left(1+\frac{\varepsilon\cdot v^{\frac1d}}{M^{1/d}}\right)^d \cdot m_{out}\left(\bigcup_{Y\in\S_n}Y\right) \intertext{where the above inequality is due the the countable subaddativity property of outer measures which states that a countable sum of outer measures of sets is at least as large as the outer measure of the union of the sets. In the last step we get}
    &= \left(1+\frac{\varepsilon\cdot v^{\frac1d}}{M^{1/d}}\right)^d \cdot m\left(S_n\right) \tag{$\bigsqcup_{Y\in\F_n}Y=S_n$ is measurable}
\end{align*}


Thus, regardless of whether $\S_n$ has infinite or finite cardinality, there exists a finite subfamily $\F_n\subseteq\S_n$ such that
\[
\sum_{Y\in\F_n}m\left(Y+\ballo{\varepsilon}{\vec{0}}\right) \geq \left(1+\frac{\varepsilon\cdot v^{\frac1d}}{M^{1/d}}\right)^d\cdot m(S_n).
\]
Fix such a subfamily $\F_n$, and let $\A_n=\set{Y+\ballo{\varepsilon}{\vec{0}}}_{Y\in\F_n}$ be a family indexed\footnote{
We require this to be an indexed family rather than just a set, because it is possible that there are distinct $Y,Y'\in\S_n$ such that $Y+\ballo{\varepsilon}{\vec{0}}=Y'+\ballo{\varepsilon}{\vec{0}}$.
} by $\F_n$. Note that for each $A_Y\defeq Y+\ballo{\varepsilon}{\vec{0}}\in\A_n$, since  $Y\subseteq S_n=\ballo{n}{\vec{0}}$, we have $A_Y\subseteq S_n+\ballo{\varepsilon}{\vec{0}}=S_n'$. Thus, by \Autoref{lower-bound-cover-number}\footnote{
We are taking $X$ in \Autoref{lower-bound-cover-number} to be $S_n'$ in this proof, and taking $\mu$ to be $m$ and taking $\A$ to be $\A_n$. We have that $\sum_{A\in\A_n}m(A)<\infty$ because $\abs{\A_n}=\abs{\F_n}$ is finite, and each $A\in\A_n$ is a subset of $S_n'$, so has finite measure.
}, there is a point $\vec{p}^{(n)}\in S_n'$ which belongs to at least $k_n$-many sets in $\A_n$ where
\begin{align*}
    k_n\defeq\ceil*{\frac{\sum_{Y\in\F_n}m\left(Y+\ballo{\varepsilon}{\vec{0}}\right)}{m(S_n')}} &\geq \frac{\left(1+\frac{\varepsilon\cdot v^{\frac1d}}{M^{1/d}}\right)^d \cdot m\left(S_n\right)}{m(S_n')} \tag{Above}\\
    &= \left(1+\frac{\varepsilon\cdot v^{\frac1d}}{M^{1/d}}\right)^d \cdot \frac{m\left(\ballo{n}{\vec{0}}\right)}{m\left(\ballo{n+\varepsilon}{\vec{0}}\right)} \tag{Def'n of $S_n$ and $S_n'$}\\
    &= \left(1+\frac{\varepsilon\cdot v^{\frac1d}}{M^{1/d}}\right)^d \cdot \frac{n^d\cdot v}{(n+\varepsilon)^d\cdot v} \tag{Standard scaling fact}\\
    &= \left(1+\frac{\varepsilon\cdot v^{\frac1d}}{M^{1/d}}\right)^d \cdot \left(\frac{n}{n+\varepsilon}\right)^d \tag{Simplify}.
\end{align*}
Since $\vec{p}^{(n)}$ belongs to at least $k_n$-many sets in $\A_n$, this by definition means that there are at least $k_n$-many sets $Y\in\F_n$ such that $\vec{p}^{(n)}\in Y+\ballo{\varepsilon}{\vec{0}}$, so by \Autoref{antisymmetry-containment-pre-claim-minkowski-sum}, we have $Y\cap\ballo{\varepsilon}{\vec{p}^{(n)}}\not=\emptyset$. For each such $Y$ (since $Y\in\F_n\subseteq \S_n)$, there is a distinct\footnote{I.e. for $Y\not=Y'\in\S_n$ we have that $X_Y,X_{Y'}\in\P$ and $X_Y\not=X_{Y'}$ so this mapping of $Y$'s to $X$'s is injective, so we have at least the same cardinality of $X$'s with the desired property as $Y$'s with the desired property.} $X_Y\in\P$ such that $Y\subseteq X_Y$ and thus $X_Y\cap\ballo{\varepsilon}{\vec{0}}\not=\emptyset$ showing that there are at least $k_n$-many sets in $\P$ which intersect $\ballo{\epsilon}{\vec{0}}$.

For the last step of the proof, we perform a limiting process on $n$. By \Autoref{smaller-ceiling}, let $\gamma\in\R$ such that $\gamma<\left(1+\frac{\varepsilon\cdot v^{\frac1d}}{M^{1/d}}\right)^d$ and $\ceil{\gamma}=\ceil*{\left(1+\frac{\varepsilon\cdot v^{\frac1d}}{M^{1/d}}\right)^d}$. Then, because 
\[
\lim_{n\to\infty} k_n \geq
\lim_{n\to\infty}\left(1+\frac{\varepsilon\cdot v^{\frac1d}}{M^{1/d}}\right)^d \cdot \left(\frac{n}{n+\varepsilon}\right)^d = \left(1+\frac{\varepsilon\cdot v^{\frac1d}}{M^{1/d}}\right)^d > \gamma,
\]
we can take $N\in\N$ sufficiently large so that
\[
k_N\geq\left(1+\frac{\varepsilon\cdot v^{\frac1d}}{M^{1/d}}\right)^d \cdot \left(\frac{N}{N+\varepsilon}\right)^d > \gamma.
\]
Then considering the point $\vec{p}^{(N)}$, we have by the choice of $\gamma$ and the fact that $k_N$ is an integer that
\[
k_N = \ceil{k_N} \geq \ceil{\gamma} = \ceil*{\left(1+\frac{\varepsilon\cdot v^{\frac1d}}{M^{1/d}}\right)^d}
\]
showing that $\ballo{\epsilon}{\vec{p}^{(N)}}$ intersects at least $k_N\geq \ceil*{\left(1+\frac{\varepsilon\cdot v^{\frac1d}}{M^{1/d}}\right)^d}$ members of $\P$ as claimed which completes the proof.
\end{proof}

\begin{remark}
\renewcommand{\A}{\mathcal{A}}
\newcommand{\F}{\mathcal{F}}
If one carefully follows the proof above with minor modification, one can show that if there is any bounded set $S\subseteq\R^d$ that intersects infinitely many elements of $\P$, then for each $i\in\N_0$, there is a point $\vec{p}^{(i)}\in\R^d$ such that $\ballostd{\epsilon}{\vec{p}^{(i)}}$ intersects at least $i$ members of $\P$. To see this, take $n_i\in\N$ large enough so that $S_{n_i}\supseteq S$, and then note that $\S_{n_i}$ has infinite cardinality, so take$\F_{n_i}$ to be a large enough subfamily of $\S_{n_i}$ that $\abs{\F_{n_i}}\cdot(\varepsilon\cdot v^{\frac1d})^d$) exceeds $i\cdot m(S_{n_i}')$ so $\sum_{Y\in\F_{n_i}}m\left(Y+\ballostd{\varepsilon}{\vec{0}}\right)\geq \abs{\F_{n_i}}\cdot(\varepsilon\cdot v^{\frac1d})^d$) exceeds $i\cdot m(S_{n_i}')$, so there is some point $\vec{p}^{(i)}$ that belongs to $Y+\ballostd{\varepsilon}{\vec{0}}$ for at least $i$-many sets $Y\in\S_{n_i}$ so that $\ballostd{\epsilon}{\vec{p}^{(i)}}$ intersects $X_Y$ for at least $i$-many sets $X_Y\in\P$.

However, one can do better than this. If there is some bounded set $S\subseteq\R^d$, then $\clos{S}$ is also bounded because the diameter is no larger. Consider the standard open cover $\set{\ballostd{\epsilon}{\vec{x}}}_{\vec{x}\in \clos{S}}$ of $\clos{S}$. Since $\clos{S}$ is closed and bounded, by the Heine-Borel theorem there is a finite subcover $\mathcal{C}\subseteq\set{\ballostd{\epsilon}{\vec{x}}}_{\vec{x}\in \clos{S}}$ of $\clos{S}$. Since $\clos{S}$ intersects infitely many sets in $\P$, the fact that $\mathcal{C}$ has finite cardinality implies one of the $\epsilon$-balls in $\mathcal{C}$ must intersect infinitely many members of $\P$. Thus, for each $\epsilon\in(0,\infty$, there is in fact some $\vec{p}\in\R^d$ such that $\ballostd{\epsilon}{\vec{p}}$ intersects infinitely many members of $\P$.
\end{remark}

With the main result now proven, we introduce one last tool which will allow us to convert the measure bound hypothesis of \Autoref{epsilon-measure-bounded-any-norm} to a diameter bound hypothesis instead. The isodiamteric inequality (given below) states that no bounded set has greater measure than the ball of the same diameter. Thus, an upper bound on the diameters of members (in any norm) immediately gives an upper bound on the measures of the members.

\begin{restatable}[Isodiametric Inequality]{theorem}{restatableIsodiametricIneq}\label{isodiametric-ineq}
Let $d\in\N$ and $\norm{\cdot}$ be any norm on $\R^d$. Let $A\subseteq\R^d$ be a bounded set. Then the outer Lebesgue measure of $A$ is at most the Lebesgue measure of the ball of the same diameter. That is, (in three equivalent forms) if $\diamstd(A)=D$, then
\begin{align*}
    m_{out}(A) &\leq m\left(\ballostd{D/2}{\vec{0}}\right) \\
    &= \left(\tfrac{D}{2}\right)^d\cdot m\left(\ballostd{1}{\vec{0}}\right) \\
    &= \left(\tfrac{D}{2}\right)^d\cdot \unitballmeasurestd.
\end{align*}
\end{restatable}
In some sources, this inequality is only proved for the $\ell_2$ norm using Steiner symmetrization since it is geometrically quite intuitive, but there is a more general proof for all norms using the \namefullref{generalized-brunn-minkowski-inequality}. For convenience, we offer a standard proof using the latter technique in \hyperlink{hypertarget-iso-diametric-ineq}{\Autoref*{sec:measure-theory}}.

\begin{remark}
If one uses the standard observation that for a bounded set $A$ with $\diamstd(A)=D$ and any point $\vec{a}\in A$ it holds that $A\subseteq \ballostd{D}{a}$, then one immediately obtains
\begin{align*}
    m_{out}(A) &\leq m\left(\ballostd{D}{\vec{a}}\right) \\
    &= m\left(\ballostd{D}{\vec{0}}\right) \\
    &= D^d\cdot m\left(\ballostd{1}{\vec{0}}\right) \\
    &= D^d\cdot \unitballmeasurestd
\end{align*}
so the \nameref{isodiametric-ineq} gives a factor of $2$ smaller diameter and a factor of $2^d$ smaller measure.
\end{remark}

\restatableEpsilonDiameterBoundedAnyNorm*

\begin{proof}
Let $M=\left(\frac D2\right)^d\cdot \unitballmeasurestd$. For any $X\in\P$, since $\diamstd(X)\leq D$, then by the \namefullref{isodiametric-ineq}, $m_{out}(X)\leq M$. Since every member of $\P$ has outer Lebesgue measure at most $M$, then by \Autoref{epsilon-measure-bounded-any-norm}, there exists $\vec{p}\in\R^d$ such that $\ballostd{\epsilon}{\vec{p}}$ intersects at least
\[
    \ceil*{\left(1+\varepsilon\left(\frac{\unitballmeasurestd}{M}\right)^{1/d}\right)^d} \\
    =\ceil*{\left(1+\frac{2\varepsilon}{D}\right)^d}
\]
members of $\P$ as claimed.
\end{proof}

We also have that \Autoref{k-exponential-in-dimension-measure} follows as a simple corollary of \Autoref{epsilon-measure-bounded-any-norm}.

\restatableKExponentialInDimensionMeasure*
\begin{proof}
Consider the $\ell_\infty$ norm and $M=1$ noting that for each 
$d\in\N$, ${\unitballmeasureinf}=2^d$ (i.e. the volume of the $\ell_\infty$ unit ball in $\R^d$ is $2^d$).

Then by \Autoref{epsilon-measure-bounded-any-norm}, there is a point $\vec{p}\in\R^d$ such that $\balloinf{\epsilon}{\vec{p}}$ intersects at least
\[
\left(1+\varepsilon\left(\frac{v_{\norm{\cdot}_\infty,d}}{M}\right)^{1/d}\right)^d = \left(1+\varepsilon\left(\frac{2^d}{1}\right)^{1/d}\right)^d = (1+2\varepsilon)^d
\]
members of $\P$, and thus trivially the closed ball $\ballcinf{\epsilon}{\vec{p}}$ does as well. Thus, if $\P$ is $(k,\epsilon)$-secluded (meaning by definition that for every $\vec{p}\in\R^d$ it holds that $\ballcinf{\epsilon}{\vec{p}}$ intersects at most $k$ members of $\P$) then it must be that $k\geq(1+2\epsilon)^d$.

For the consequence, if $k\leq 2^d$, then this implies $\epsilon\leq\frac12$. Then taking the logarithm of both sides of our inequality and using the fact that $\ln(1+x)\geq\frac x2$ for $x\in[0,1]$, we have
\[
\ln(k) \geq d\ln(1+2\epsilon) \geq d\epsilon
\]
showing that $\epsilon\leq\frac{\ln(k)}{d}$.
\end{proof}

We state the following interesting corollary when $k(d)$ is polynomial.
\begin{corollary}
Let $\langle \P_d\rangle_{d=1}^\infty$ be a sequence of $(k(d), \varepsilon(d))$-secluded partitions of $\R^d$ such that every member of each $\P_d$ has outer Lebesgue measure at most 1. If $k(d) \in \poly(d)$ then $\varepsilon(d) \in O(\frac{\ln d}{d})$ (where the hidden constant can be taken to be anything exceeding the polynomial degree of $k$).
\end{corollary}

\begin{proof}
Since $k(d)\in\poly(d)$, then there are constants $C,n$ such that for sufficiently large $d$, we have $k(d)\leq Cd^n$ which for sufficiently large $d$ is less than $2^d$ so by \Autoref{k-exponential-in-dimension-measure}, for sufficiently large $d$ we have
\[
\epsilon(d)\leq\frac{\ln(k(d))}{d}\leq\frac{n\ln(Cd)}{d}\in O\left(\frac{\ln(d)}{d}\right).
\]

More specifically, for any $n'>n$ we have for large enough $d$ that $(n'-n)\ln(d)\geq n\ln(C)$, so for large enough $d$ we have
\[
\epsilon(d)
\leq \frac{n\ln(Cd)}{d}
= \frac{n[\ln(C)+\ln(d)]}{d}
\leq \frac{(n'-n)\ln(d)+n\ln(d)}{d}
= \frac{n'\ln(d)}{d}
\]
showing that the hidden constant can be taken to be any $n'$ larger than the degree $n$ of $k$.
\end{proof}

\subsection{Specific Norms}


While we won't do an elaborate analysis with any specific norms, we will at least mention how \Autoref{epsilon-measure-bounded-any-norm} evaluates when we consider some of the most common norms. For the $\ell_1$, $\ell_2$, and $\ell_\infty$ norms, the value of $v_d$ is respectively $\frac{2^d}{d!}$, $\frac{\pi^{d/2}}{\Gamma\left(\frac d2+1\right)}$, and $2^d$. Here, $\Gamma$ denotes the gamma-function which generalizes the factorial (specifically, for natural numbers $n$, $\Gamma(n+1)=n!$). The volumes of the $\ell_2$ and $\ell_\infty$ ball are well-known, and the volume of the $\ell_1$ ball can be be obtained from \cite{wang_volumes_2005}, or one can recognize that the $\ell_1$ unit ball is (disregarding boundaries) a disjoint union of $2^d$ copies of the standard simplex---one in each orthant, and each simplex has measure $\frac{1}{d!}$.

\begin{table}[h]
    \centering
    $\begin{array}{|c|ll|}\hline
    \text{Norm} & \text{Lower Bound}\\\hline
    \ell_1 & \left(1+\frac{2\varepsilon}{\left(M\cdot(d!)\right)^{1/d}}\right)^d &\approx \left(1+\varepsilon\frac{2 e}{M^{1/d}\cdot d}\right)^d \\\hline
    \ell_2 & \left(1+\frac{\varepsilon\sqrt{\pi}}{\left(M\cdot \Gamma\left(\frac d2+1\right)\right)^{1/d}}\right)^d &\approx \left(1+\varepsilon\frac{\sqrt{2\pi e}}{M^{1/d}\cdot \sqrt{d}}\right)^d \\\hline
    \ell_\infty & \left(1+\varepsilon\frac{2}{M^{1/d}}\right)^d & \\\hline
    \end{array}$
    \caption{Lower bounds evaluated for some common norms. The approximate value for $\ell_1$ uses Stirling's approximation to say that $(d!)^{1/d}\approx \left(\sqrt{2\pi d}\left(\frac d e\right)^d\right)^{1/d}\approx \frac de$. The approximate value for $\ell_2$ uses Stirling's approximation and $\Gamma\left(\frac d2+1\right)\approx \ceil{\frac d2}!$ to get $\left(\Gamma\left(\frac d2+1\right)\right)^{1/d} \approx \left(\ceil{\frac d2}!\right)^{1/d} \approx \left(\sqrt{2\pi \ceil{\frac d2}}\left(\frac {\ceil{\frac d2}} e\right)^{\ceil{\frac d2}}\right)^{1/d} \approx \left(\left(\frac {\ceil{\frac d2}} e\right)^{\ceil{\frac d2}}\right)^{1/d} \approx \sqrt{\frac{d}{2e}}$.}
    \label{tab:common-norm-lower-bounds}
\end{table}

The lower bounds on $k(d)$ based on $\epsilon(d)$ stated in \Autoref{epsilon-measure-bounded-any-norm} for these three specific norms can be found in \Autoref{tab:common-norm-lower-bounds}. The main observation that we want to make is that when using norms other than $\ell_\infty$, the factor of of $\varepsilon$ is no longer a constant, but a decreasing function of the dimension. This should not be too surprising since for $\ell_p$ norms which are not $\ell_\infty$, the unit $\ell_p$-ball is a subset of the unit $\ell_\infty$ ball, so the measure will be smaller and actually tend to $0$ as $d$ tends to $\infty$, so we should expect to not be able to intersect as many members of the partition.

%% file: 5_sperner_kkm_lebesgue.tex
\section{A Neighborhood Sperner/KKM/Lebesgue Theorem}\label{sec:NewSperner}

In this section, we restate and prove our neighborhood variant of the Sperner/KKM/Lebesgue result on the cube. The proof is illustrated in \Autoref{fig:sperner-kkm-coloring}.

\input{images/sperner_kkm_lebesgue/sperner_kkm_lebesgue_figure}

\restatableNewSpernerKKMLebesgue*
\begin{proof}
\renewcommand{\A}{\mathcal{A}}
For convenience, we will assume that the cube is $[-\frac12,\frac12]^d$ rather than $[0,1]^d$. Let $C$ be a set (of colors) and $\chi\colon[-\frac12,\frac12]^d\to C$ be such a coloring of the unit cube $[-\frac12,\frac12]^d$. Note that if $C$ has infinite cardinality then because we can cover the cube with finitely many $\epsilon$-balls one of them must intersect infinitely many color sets so the result is true. Thus we assume that $C$ has finite cardinality.

For each color $c\in C$ we will let $X_c$ denote the set of points assigned color $c$ by $\chi$---that is, $X_c=\chi^{-1}(c)$. Note that the hypothesis that no color includes points of opposite faces formally means that for every color $c\in C$, the set $X_c$ has the property that for each coordinate $i\in[d]$, the projection $\pi_i(X_c)=\set{x_i:\vec{x}\in X_c}$ does not contain both $-\frac12$ and $\frac12$.

The first step in the proof is to extend the coloring $\chi$ to the larger cube $[-\frac12-\epsilon,\frac12+\epsilon]^d$ in a natural way. Consider the following function $f$ which truncates points in the larger interval to be in the smaller interval:
\begin{align*}
&f\colon[-\tfrac12-\epsilon,\tfrac12+\epsilon]\to[-\tfrac12,\tfrac12] \\
&f(y)\defeq\begin{cases}
-\frac12 & y\leq-\frac12 \\
y        & y\in(-\frac12,\frac12) \\
\frac12 & y\geq\frac12
\end{cases}
\end{align*}
Let $\vec{f}\colon[-\tfrac12-\epsilon,\tfrac12+\epsilon]^d\to[-\tfrac12,\tfrac12]^d$ be the function which is $f$ in each coordinate: $\vec{f}(\vec{y})\defeq\langle f(y_i)\rangle_{i=1}^d$.

Now extend the coloring $\chi$ to the coloring $\gamma\colon[-\frac12-\epsilon,\frac12+\epsilon]^d\to C$ defined by
\[
\gamma(\vec{x}) \defeq \chi\left(\vec{f}\left(\vec{x}\right)\right).
\]

For each color $c\in C$, let $Y_c=\gamma^{-1}(c)$ denote the points assigned color $c$ by $\gamma$ and note that $Y_{c}\subseteq X_{c}$. Consistent with this notation, we will typically refer to a point in the unit cube as $\vec{x}$ and a point in the extended cube as $\vec{y}$.

We make the following claim which implies that for each color $c\in C$, the set $Y_c$ of points of that color in the extended coloring are contained in a set bounded away from one side of the extended cube $[-\frac12-\epsilon,\frac12+\epsilon]^d$ in each coordinate.

\begin{subclaim}\label{subclaim-1}
For each color $c\in C$ there exists an orientation $\vec{v}\in\set{-1,1}^d$ such that $Y_c\subseteq\prod_{i=1}^d v_i\cdot(-\tfrac12,\tfrac12+\epsilon]$.
\end{subclaim}
\begin{subclaimproof}
Fix an arbitrary coordinate $i\in[d]$. Note that for every $\vec{y}\in Y_c$ we have $f(\vec{y}\in X_c$ which is to say that the $\vec{y}$ has the same color in the extended coloring as $f(\vec{y})$ does in the original coloring (see justification\footnote{\label{footnote:matching-colors}
    For every $\vec{y}\in Y_c$ we have (by definition of $Y_c$) that $\gamma(\vec{y})=c$ and (by definition of $\gamma$) that $\gamma(\vec{y})=\chi(f(\vec{y}))$ showing that $\chi(f(\vec{y}))=c$ and thus (by definition of $X_c$) that $f(\vec{y})\in X_c$. 
}).

Note that if there is some $\vec{y}\in Y_c$ with $y_i\leq-\frac12$, then $f(y_i)=-\frac12$ so $\pi_i(X_c)\ni f(y_i)=-\frac12$. Similarly, if there is some $\vec{y}\in Y_c$ with $y_i\geq\frac12$, then $\pi_i(X_c)\ni\frac12$. Recall that by hypothesis, $\pi_i(X_c)$ does not contain both $-\frac12$ and $\frac12$ which means it is either the case that for all $\vec{y}\in Y_c$ we have $y_i>-\frac12$ (so $\pi_i(Y_c)\subseteq(-\frac12,\frac12+\epsilon]$) or it is the case that for all $\vec{y}\in Y_c$ we have $y_i<\frac12$ (so $\pi_i(Y_c)\subseteq[-\frac12-\epsilon,\frac12)$).

Thus we can choose $v_i\in\set{-1,1}$ such that $\pi_i(Y_c)\subseteq v_i\cdot(-\frac12,\frac12+\epsilon]$. Since this is true for each coordinate $i\in[d]$ we can select $\vec{v}\in\set{-1,1}^d$ such that
\[
Y_c \subseteq \prod_{i=1}^d\pi_i(Y_c) \subseteq \prod_{i=1}^d v_i\cdot(-\frac12,\frac12+\epsilon]
\]
as claimed.
\end{subclaimproof}

For an orientation $\vec{v}\in\set{-1,1}^d$, let $B_{\vec{c}}$ denote the set $B_{\vec{v}}\defeq\prod_{i=1}^d -v_i\cdot(0,\epsilon)$ which should be interpreted as an open orthant of the $\ell_\infty$ $\epsilon$-ball centered at the origin---specifically the orthant opposite the orientation $\vec{v}$. Building on \Autoref{subclaim-1}, we get the following:

\begin{subclaim}\label{subclaim-2}
For each color $c\in C$, there exists an orientation $\vec{v}\in\set{-1,1}^d$ such that $Y_c+B_{\vec{v}} \subseteq [-\tfrac12-\epsilon,\tfrac12+\epsilon]^d$.
\end{subclaim}
\begin{subclaimproof}
Let $\vec{v}$ be an orientation given in \Autoref{subclaim-1} for color $c$. We get the following chain of containments:
\begin{align*}
Y_c+B_{\vec{v}} &= Y_c + \left(\prod_{i=1}^d -v_i\cdot(0,\epsilon)\right) \tag{Def'n of $B_{\vec{v}}$} \\
&\subseteq \left( \prod_{i=1}^d v_i\cdot(-\tfrac12,\tfrac12+\epsilon] \right) + \left(\prod_{i=1}^d -v_i\cdot(0,\epsilon)\right) \tag{\Autoref{subclaim-1}} \\
&= \left( \prod_{i=1}^d v_i\cdot(-\tfrac12,\tfrac12+\epsilon] \right) + \left(\prod_{i=1}^d v_i\cdot(-\epsilon,0)\right) \tag{Factor a negative} \\
&= \prod_{i=1}^d v_i\cdot(-\tfrac12-\epsilon,\tfrac12+\epsilon) \tag{Minkowski sum of rectangles} \\
%
&\subseteq [-\tfrac12-\epsilon,\tfrac12+\epsilon]^d \tag{$v_i\in\set{-1,1}$}.
\end{align*}
This proves the claim.
\end{subclaimproof}

We also claim that $Y_c+B_{\vec{v}}$ has a substantial measure.

\begin{subclaim}\label{subclaim-3}
For each color $c\in C$ and any orientation $\vec{v}\in\set{-1,1}^d$, the set $Y_c+B_{\vec{v}}$ is Borel measurable and $m(Y_c+B_{\vec{v}})\geq m_{out}(Y_c) \cdot \left(1+\frac{\epsilon}{1+\epsilon}\right)^d$.
\end{subclaim}
\begin{subclaimproof}
Let $M=(1+\epsilon)^d$ which is the measure of $\prod_{i=1}^d v_i\cdot(-\tfrac12,\tfrac12+\epsilon]$, and because by \Autoref{subclaim-1}, $Y_c$ is a subset of this set, we have $m_{out}(Y_c)\leq M$.

We have that $Y_c+B_{\vec{v}}$ is Borel measurable and that $m\left(Y_c+B_{\vec{v}}\right)\geq \left(m_{out}(Y_c)^\frac1d+\epsilon\right)^d$ by \Autoref{outer-measure-brunn-minkowsi-bound} (see details\footnote{
    Note that for the $\ell_\infty$ norm, the measure of the unit ball is $v_{\norm{\cdot}_\infty,d}=2^d$. Then note that $B_{\vec{v}}$ is an open orthant of an $\epsilon$ ball with respect to $\ell_\infty$, so is in fact itself an $\frac\epsilon2$ ball with respect to $\ell_\infty$. This is why we get ``$\epsilon$'' instead of the $``2\epsilon$'' in \Autoref{outer-measure-brunn-minkowsi-bound}. We could translate this open ball to the origin and translate the set $Y_c$ accordingly to get the same Minkowski sum without changing the measures, and after doing so we could apply  \Autoref{outer-measure-brunn-minkowsi-bound} verbatim.
}). Thus, we have the following chain of inequalities:
\begin{align*}
    m(Y_c+B_{\vec{v}}) &\geq \left(m_{out}(Y_c)^{1/d} + \epsilon\right)^d \tag{Above}\\
    &= M\cdot \left(\frac{m_{out}(Y_c)^{1/d}}{M^{1/d}} + \frac{\epsilon}{M^{1/d}}\right)^d \tag{Factor out $M$} \\
    &\geq M\cdot \left(\frac{m_{out}(Y_c)}{M}\right) \cdot \left(1+\frac{\epsilon}{M^{1/d}}\right)^d \tag{\Autoref{brunn-minkowski-bound-lemma}} \\
    &= m_{out}(Y_c) \cdot \left(1+\frac{\epsilon}{1+\epsilon}\right)^d \tag{Simplify and use $M=(1+\epsilon)^d$}
\end{align*}
\end{subclaimproof}

Now, consider the indexed family $\A=\set{Y_c+B_{\vec{v}(c)}}_{c\in C}$ (where $\vec{v}(c)$ is an orientation for $c$ as in \Autoref{subclaim-1} and \Autoref{subclaim-2}) noting that it has finite cardinality because $C$ has finite cardinality. Considering the sum of measures of sets in $\A$, we have the following:
\begin{align*}
    \sum_{A\in\A}m(A) &= \sum_{c\in C}m\left(Y_c+B_{\vec{v}(c)}\right) \tag{Def'n of $\A$; measurability was shown above}\\
    &\geq \left(1+\frac{\epsilon}{1+\epsilon}\right)^d\cdot\sum_{c\in C}m_{out}(Y_c) \tag{\Autoref{subclaim-3} and linearity of summation}\\
    &\geq \left(1+\frac{\epsilon}{1+\epsilon}\right)^d\cdot m_{out}\left(\bigcup_{c\in C}Y_c\right) \tag{Countable/finite subaddativity of outer measures}\\
    &= \left(1+\frac{\epsilon}{1+\epsilon}\right)^d\cdot m_{out}\left([-\tfrac12-\epsilon,\tfrac12+\epsilon]^d\right) \tag{The $Y_c$'s partition $[-\tfrac12-\epsilon,\tfrac12+\epsilon]^d$}\\
    &= \left(1+\frac{\epsilon}{1+\epsilon}\right)^d\cdot (1+2\epsilon)^d \tag{Evaluate outer measure}\\
\end{align*}

By \Autoref{subclaim-2}, each member of $\A$ is a subset of $[-\frac12-\epsilon,\frac12+\epsilon]^d$, so by \Autoref{lower-bound-cover-number-Rd}, there exists a point $\vec{p}\in[-\tfrac12-\epsilon,\tfrac12+\epsilon]^d$ that belongs to at least
\[
\ceil*{\frac{\left(1+\frac{\epsilon}{1+\epsilon}\right)^d\cdot (1+2\epsilon)^d}{(1+2\epsilon)^d}} = \ceil*{\left(1+\frac{\epsilon}{1+\epsilon}\right)^d}
\]
sets in $\A$. That is, $\vec{p}$ belongs to $Y_c+B_{\vec{v}(c)}$ for at least $\ceil*{\left(1+\frac{\epsilon}{1+\epsilon}\right)^d}$ colors $c\in C$. For each such color $c$, it follows that $\vec{p}+(-\epsilon,\epsilon)^d$ intersects $Y_c$ (see justification\footnote{
    If $\vec{p}\in Y_c+B_{\vec{v}(c)} \subseteq Y_c + (-\epsilon,\epsilon)^d$, then by definition of Minkowski sum there exists $\vec{y}\in Y_c$ and $\vec{b}\in (-\epsilon,\epsilon)^d$ such that $\vec{p}=\vec{y}+\vec{b}$ so $Y_c\ni \vec{y}=\vec{p}-\vec{b}\in \vec{p}+(-\epsilon,\epsilon)^d$ demonstrating that these two sets contain a common point.
}). Note that with respect to the $\ell_\infty$ norm,  $\vec{p}+(-\epsilon,\epsilon)^d=\balloinf{\epsilon}{\vec{p}}$ showing that $\balloinf{\epsilon}{\vec{p}}$ contains points of at least $\ceil*{\left(1+\frac{\epsilon}{1+\epsilon}\right)^d}$ colors (according to the coloring of $\gamma$ since we are discussing sets $Y_c$).

What we really want, though, is a point in the unit cube that has this property rather than a point in the extended cube, and we want it with respect to the original coloring $\chi$ rather than the extended coloring $\gamma$. We will show that the point $\vec{f}\left(\vec{p}\right)$ suffices.

\begin{subclaim}\label{subclaim-4}
If $c\in C$ is a color for which $\balloinf{\epsilon}{\vec{p}}\cap Y_c \not=\emptyset$, then also $\balloinf{\epsilon}{\vec{f}\left(\vec{p}\right)}\cap X_c \not=\emptyset$.
\end{subclaim}
\begin{subclaimproof}
Let $\vec{y}\in \balloinf{\epsilon}{\vec{p}}\cap Y_{\vec{c}}$. Then because $\vec{y}\in \balloinf{\epsilon}{\vec{p}}$, we have $\norm{\vec{y}-\vec{p}}_\infty<\epsilon$, so for each coordinate $i\in[d]$, $\abs{y_i-p_i}<\epsilon$. It is easy to analyze the $9$ cases (or $3$ by symmetries) arising in the definition of $f$ to see that this implies $\abs{f(y_i)-f(p_i)}<\epsilon$ as well (i.e. $f$ maps pairs of values in its domain so that they are no farther apart), thus $\norm{\vec{f}\left(\vec{y}\right)-\vec{f}\left(\vec{p}\right)}_\infty<\epsilon$ and thus $\vec{f}\left(\vec{y}\right)\in \balloinf{\epsilon}{\vec{f}\left(\vec{p}\right)}$.

Also, as justified in a prior footnote$^{10}$, for any $\vec{y}\in Y_c$ we have $\vec{f}(\vec{y})\in X_c$ so that $\vec{f}\left(\vec{y}\right)\in\balloinf{\epsilon}{\vec{f}\left(\vec{p}\right)}\cap X_c$ which shows that the intersection is non-empty.
\end{subclaimproof}

Thus, because $\balloinf{\epsilon}{\vec{p}}$ intersects $Y_c$ for at least $\ceil*{\left(1+\frac{\epsilon}{1+\epsilon}\right)^d}$ choices of color $c\in C$, by \Autoref{subclaim-4} $\vec{f}\left(\vec{p}\right)$ is a point in the unit cube which intersects $X_c$ for at least $\ceil*{\left(1+\frac{\epsilon}{1+\epsilon}\right)^d}$ different colors $c\in C$. That is, this ball contains points from at least this many of the original color sets.

The final step in the proof of the theorem is to clean up the expression with an inequality. Note that $C$ must contain of at least $2^d$ colors because each of the $2^d$ corners of the unit cube must be assigned a unique color since any pair of corners belong to an opposite pair of faces on the cube. For this reason it is trivial that for $\epsilon>\frac12$ there is a point $\vec{p}$ such that $\balloinf{\epsilon}{\vec{p}}$ intersects at least $2^d$ colors: just let $\vec{p}$ be the midpoint of the unit cube. Thus, the only interesting case is $\epsilon\in(0,\frac12]$, and for such $\epsilon$ we have $1+\epsilon\leq\frac32$ and thus $\frac{\epsilon}{1+\epsilon}\geq\frac23\epsilon$ showing that $\left(1+\frac{\epsilon}{1+\epsilon}\right)^d\geq (1+\frac23\epsilon)^d$. This completes the proof of the theorem.
\end{proof}

%% file: images/sperner_kkm_lebesgue/sperner_kkm_lebesgue_figure.tex
\begin{figure}[p]
\sbox1{%
    \begin{subfigure}{.3\textwidth}
        \includegraphics[width=\textwidth]{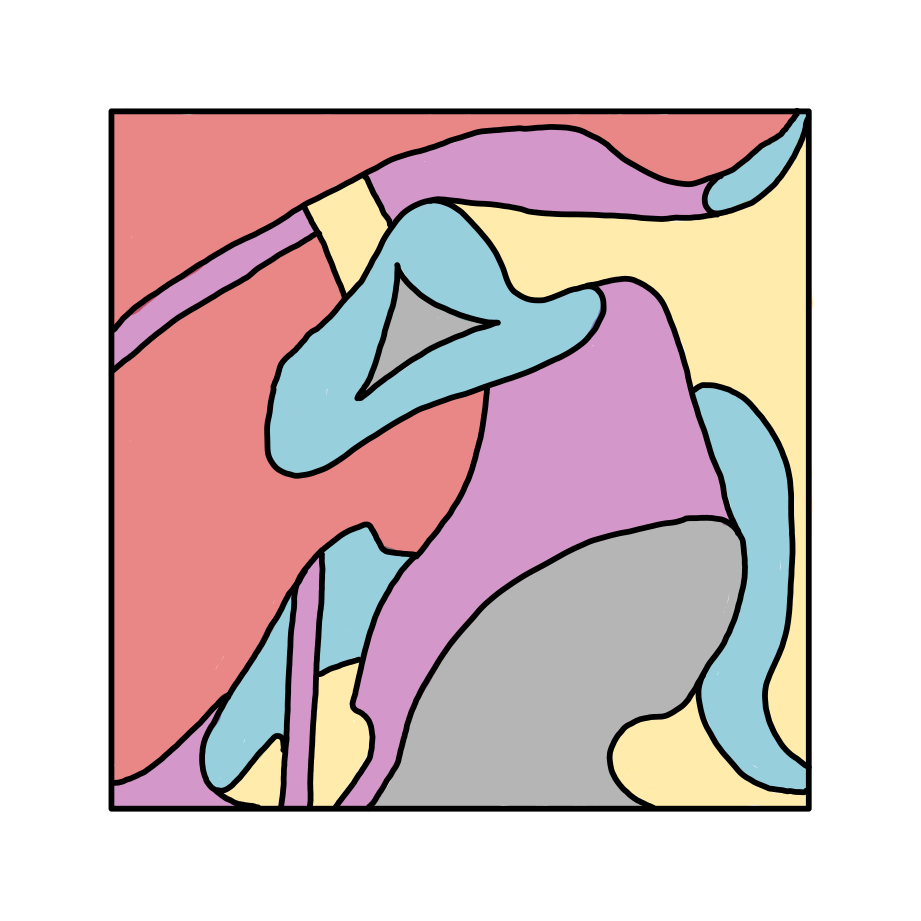}
        \subcaption{Initial coloring $\chi$}
        \label{subfig:1_non_spanning_coloring}
    \end{subfigure}
}
\sbox2{%
    \begin{subfigure}{.3\textwidth}
        \includegraphics[width=\textwidth]{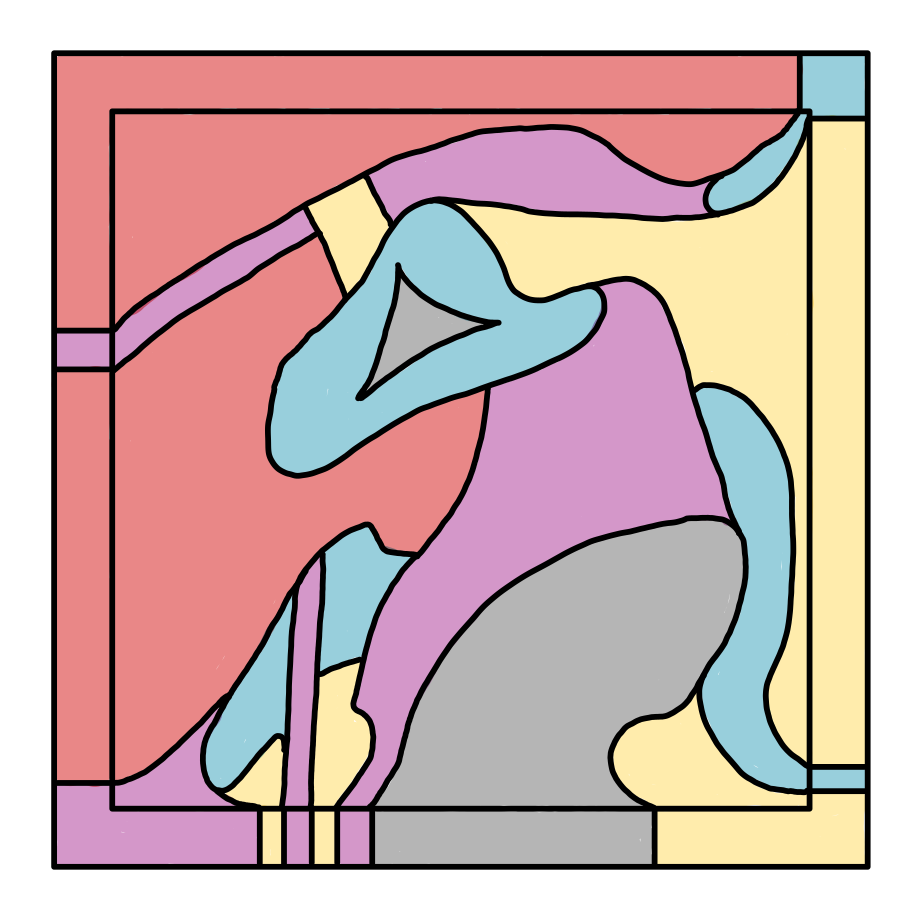}
        \subcaption{Extended coloring $\gamma$}
        \label{subfig:2_extended_coloring}
    \end{subfigure}
}
\sbox3{
    \begin{subfigure}{.3\textwidth}
        \includegraphics[width=\textwidth]{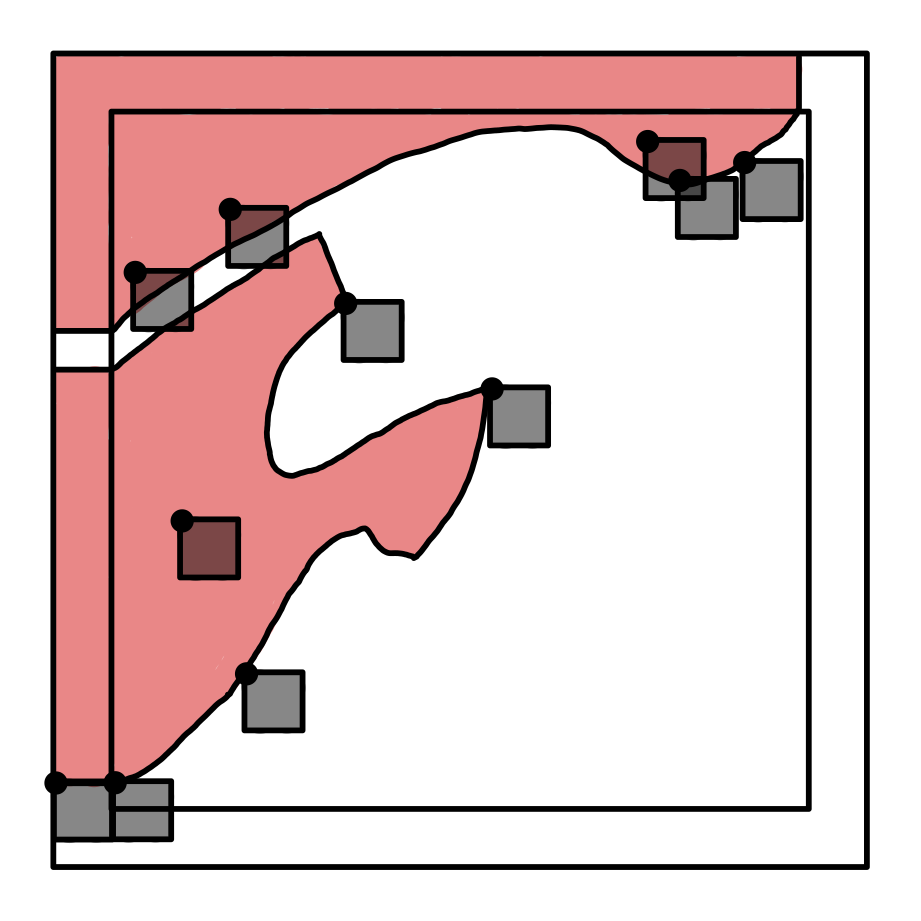}
        \subcaption{Red points ($Y_{\text{red}}$)}
        \label{subfig:3_red}
    \end{subfigure}
}
\sbox4{
    \begin{subfigure}{.3\textwidth}
        \includegraphics[width=\textwidth]{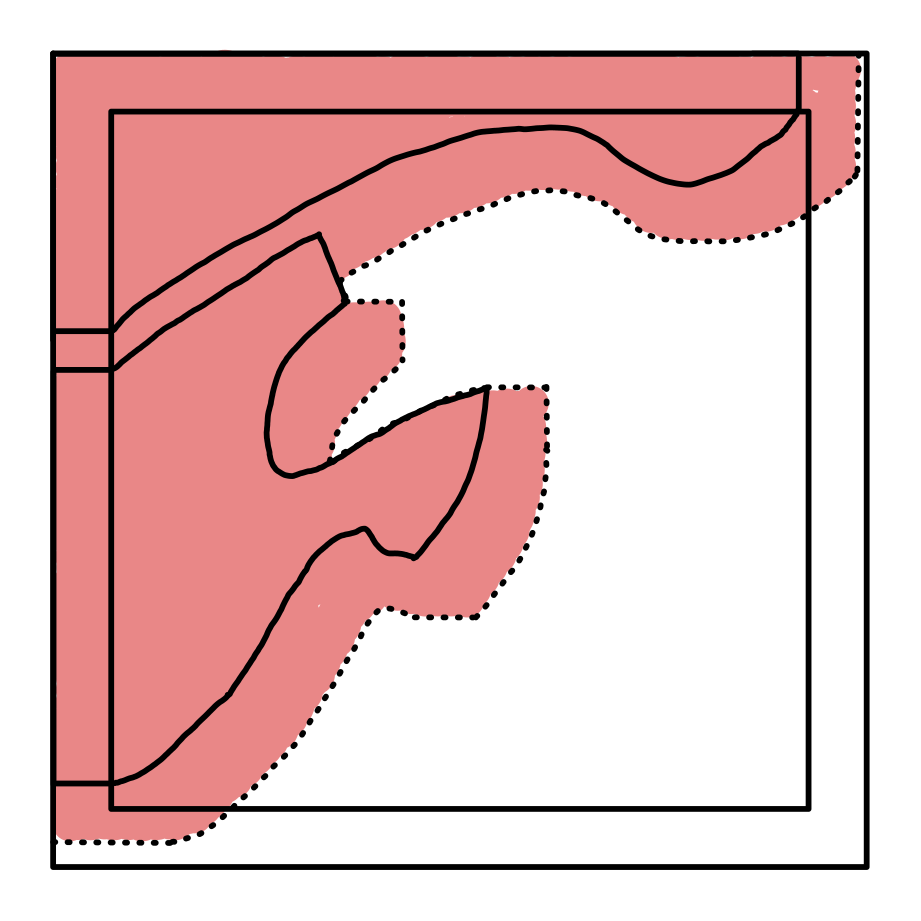}
        \subcaption{Ball added ($Y_{\text{red}}+B_{\vec{v}}$)}
        \label{subfig:4_red}
    \end{subfigure}
}
\sbox5{
    \begin{subfigure}{.3\textwidth}
        \includegraphics[width=\textwidth]{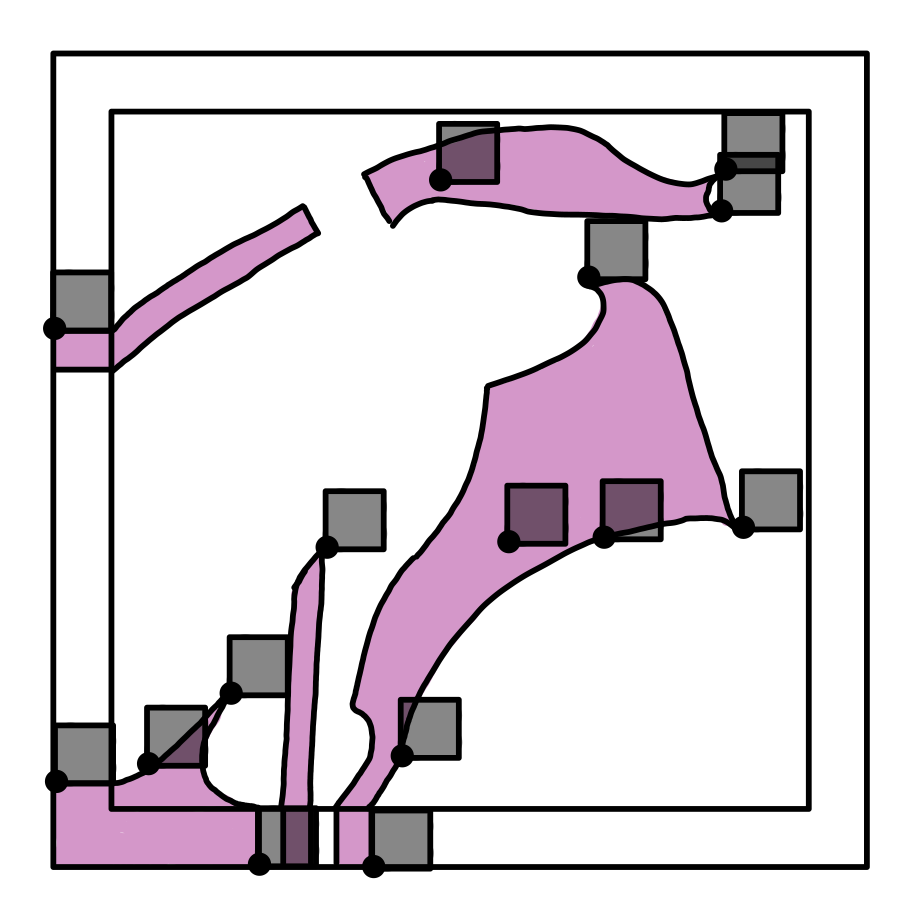}
        \subcaption{Purple points ($Y_{\text{purple}}$)}
        \label{subfig:5_purple}
    \end{subfigure}
}
\sbox6{
    \begin{subfigure}{.3\textwidth}
        \includegraphics[width=\textwidth]{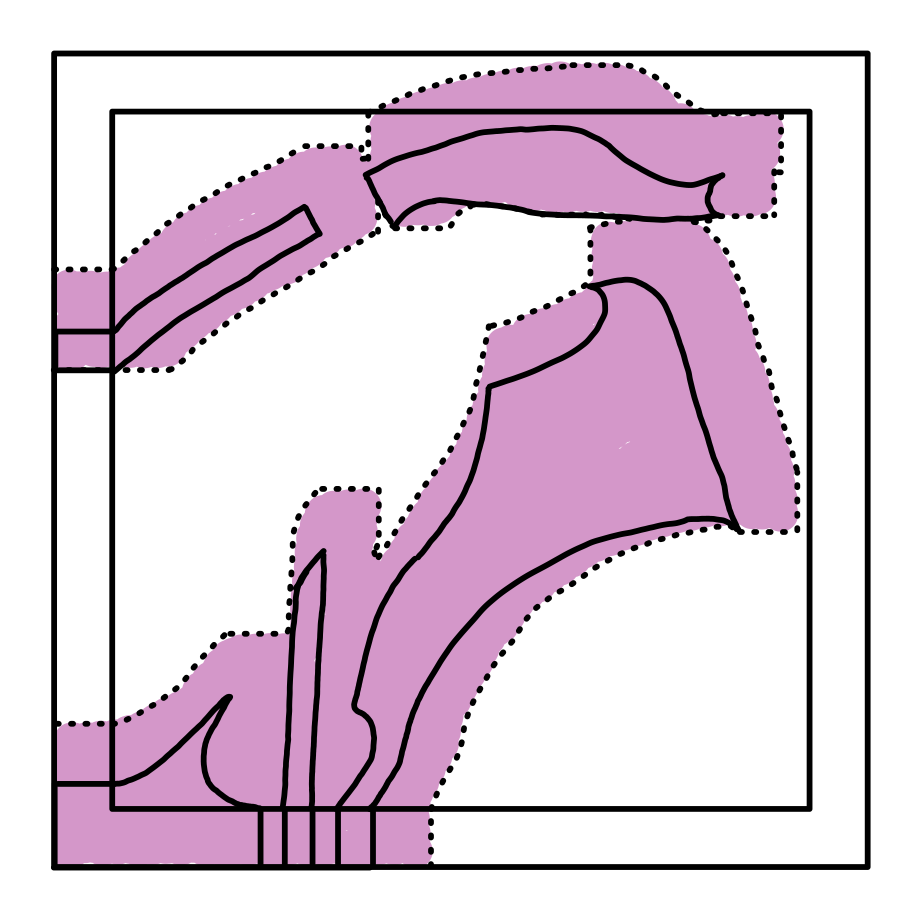}
        \subcaption{Ball added ($Y_{\text{purple}}+B_{\vec{v}}$)}
        \label{subfig:6_purple}
    \end{subfigure}
}
\sbox7{
    \begin{subfigure}{.3\textwidth}
        \includegraphics[width=\textwidth]{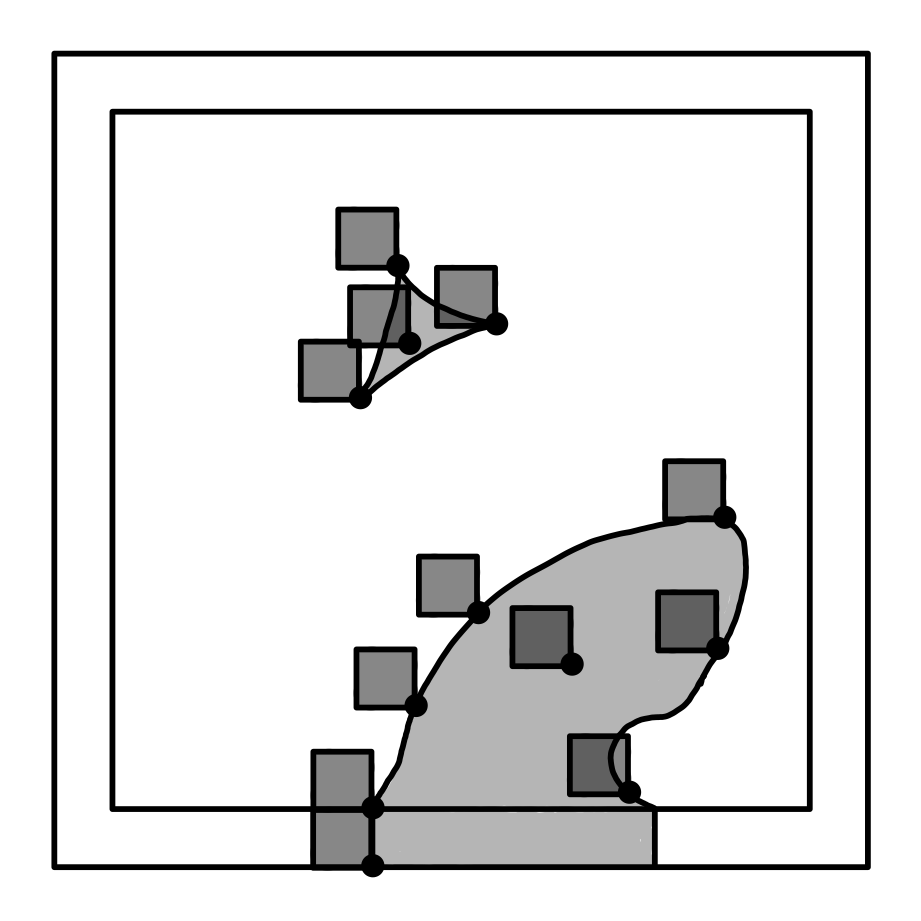}
        \subcaption{Gray points ($Y_{\text{gray}}$)}
        \label{subfig:7_gray}
    \end{subfigure}
}
\sbox8{
    \begin{subfigure}{.3\textwidth}
        \includegraphics[width=\textwidth]{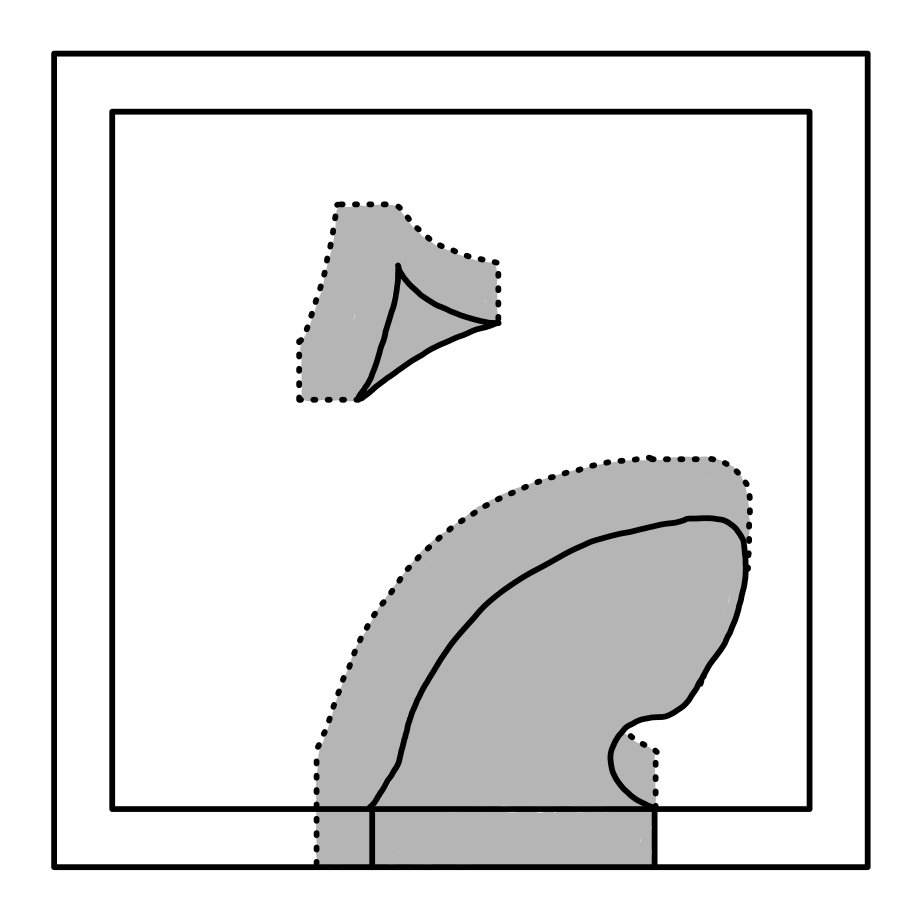}
        \subcaption{Ball added ($Y_{\text{gray}}+B_{\vec{v}}$)}
        \label{subfig:8_gray}
    \end{subfigure}
}
\sbox0{%
    \begin{subfigure}{.15\textwidth}
        \includegraphics[width=\textwidth]{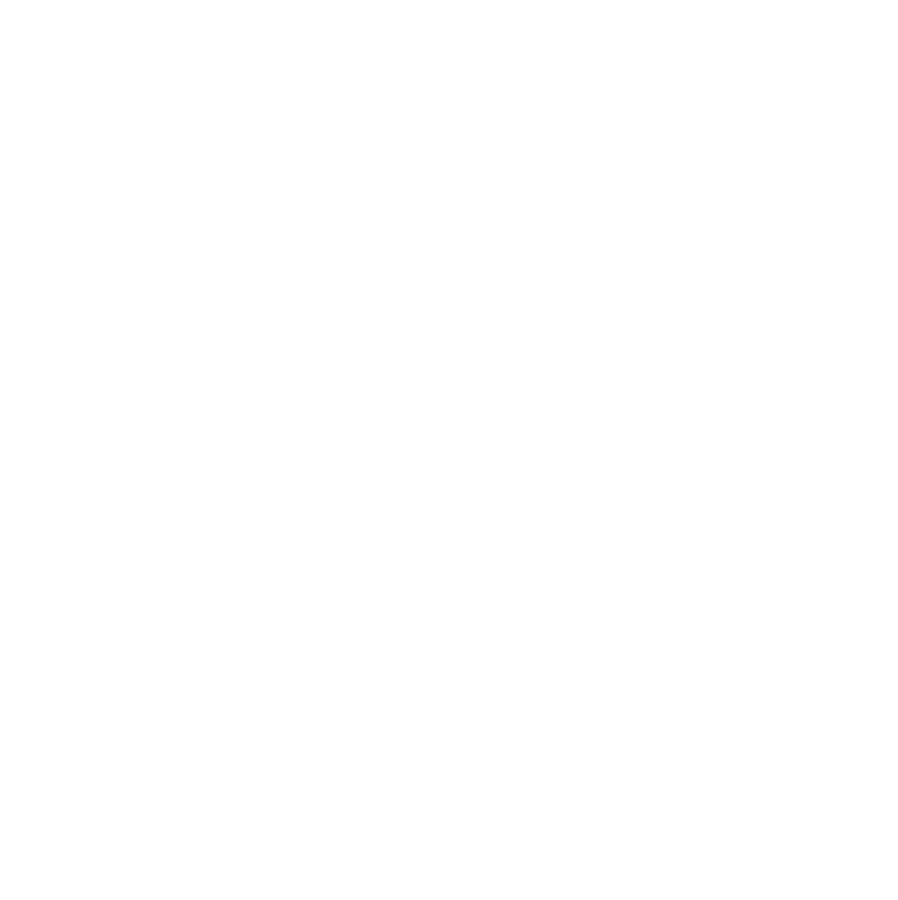}
    \end{subfigure}
}

\centering
\usebox0 \hfill \usebox1 \hfill \usebox2 \hfill \usebox0
\usebox3 \hfill \usebox5 \hfill \usebox7
\usebox4 \hfill \usebox6 \hfill \usebox8

\caption{
(\subref{subfig:1_non_spanning_coloring}) shows a (finite) coloring $\chi$ of the unit cube $[-\frac12,\frac12]^2$ such that no color includes points on opposite edges.
(\subref{subfig:2_extended_coloring}) shows the natural extension $\gamma$ of that coloring to $[-\frac12-\epsilon,\frac12+\epsilon]^2$. The extension is obtained by mapping each point $\vec{y}$ to the point $\vec{x}$ for which each coordinate value $y_i$ is restricted to be within $[-\frac12,\frac12]$, and then $\vec{y}$ is given whatever color $\vec{x}$ had.
(\subref{subfig:3_red}), (\subref{subfig:5_purple}), and (\subref{subfig:7_gray}) show three of the five colors and demonstrate that there is at least one quadrant of the $\epsilon$-ball that can be Minkowski summed with the color so that the sum remains a subset of the extended cube. For red it is the lower right quadrant, for purple it is the upper right, and for gray it could be the upper left (shown) or the upper right. (\subref{subfig:4_red}), (\subref{subfig:6_purple}), and (\subref{subfig:8_gray}) show the resulting Minkowski sum for each color. Utilizing the Brunn-Minkowski inequality, this set will have substantially greater area---by a factor of at least $(1+\frac{\epsilon}{1+\epsilon})^2$.
}

\label{fig:sperner-kkm-coloring}

\end{figure}

%% file: 6_product_partitions.tex
\section{New Constructions}\label{sec:constructions}

The partitions that we construct in this section are of a very natural form: we build a partition of a large dimension $d$ space, by splitting up the coordinates into smaller sets, and separately partitioning each set of coordinates. In the end, the smaller partitions will be known partition constructions~\cite{geometry_of_rounding}. We will define the construction very generically. We need two basic results. The following observation notes that if a partition is $(k,\varepsilon)$-secluded, then we can {\em increase} $k$ to $k'$ and {\em decrease} $\varepsilon$ to $\varepsilon'$ and the partition is trivially $(k',\varepsilon')$-secluded.
\begin{observation}[Monotonicity in $k$ and $\varepsilon$]\label{monotonicity}
Let $d\in\N$, $k,k'\in\N$ with $k'\geq k$, $\varepsilon,\varepsilon'\in[0,\infty)$ with $\varepsilon'\leq\varepsilon$, and $\P$ a $(k,\varepsilon)$-secluded partition of $\R^d$. Then $\P$ is also a $(k',\varepsilon')$-secluded partition of $\R^d$.
\end{observation}
\begin{proof}
Since $\P$ is $(k,\varepsilon)$-secluded, by definition every $\varepsilon$-ball intersects at most $k$ members of $\P$, so trivially every (no larger) $\varepsilon'$-ball intersects at most $k'\geq k$ members of $\P$.
\end{proof}

We will frequently refer to the above observation just using the phrase ``by monotonicity, $\P$ is $(k',\varepsilon')$-secluded"

\begin{fact}[Trivial $k$ for Unit Cube Partitions]\label{trivial-unit-cube-partition-seclusion}
Let $d\in\N$, $\varepsilon\in[0,\infty)$, and $\P$ be a unit cube partition of $\R^d$. Then $\P$ is $(k,\varepsilon)$-secluded for $k=\floor*{(2+2\varepsilon)^d}$.
\end{fact}
\begin{proof}
Consider any point $\vec{p}\in\R^d$. Observe that any $X\in\P$, $X$ is a unit cube, so $\diaminf(X)=1$, so if $X$ intersects $\ballcinf{\epsilon}{\vec{p}}$, then $X\subseteq\ballcinf{1+\epsilon}{\vec{p}}$.

Because (1) each $X\in\P$ has measure $1$, and (2) every pair of members are disjoint (because $\P$ is a partition), and (3) the measure of $\ballcinf{1+\varepsilon}{\vec{p}}=\vec{p}+[-1-\epsilon,1+\epsilon]^d$ is $(2+2\varepsilon)^d$, it follows that at most $\floor*{(2+2\varepsilon)^d}$ members of $\P$ are a subset of $\ballcinf{1+\epsilon}{\vec{p}}$ and thus at most $\floor*{(2+2\varepsilon)^d}$ members of $\P$ intersect $\ballcinf{\epsilon}{\vec{p}}$ which shows that $\P$ is $(k,\varepsilon)$-secluded for $k=\floor*{(2+2\varepsilon)^d}$ as claimed.
\end{proof}

\subsection{Construction}

\begin{definition}[Partition Product]\label{defn:partition-product}
Let $d_1,\ldots,d_n\in\N$ and $\P_1,\ldots,\P_n$ be partitions of $\R^{d_1},\ldots,\R^{d_n}$ respectively. Letting $d=\sum_{i=1}^{n}d_n$ we define the product partition of $\R^d$ as
\[
    \prod_{i=1}^{n}\P_i\defeq\set{\prod_{i=1}^{n}X_i\colon X_i\in\P_i}
\]
where $\prod_{i=1}^{n}X_i$ is viewed as a subset of $\R^d$.
\end{definition}

We specifically stated that $\prod_{i=1}^{n}X_i$ is viewed as a subset of $\R^d$, because technically it is a subset of $\prod_{i=1}^n \R^{d_i}$, but this is naturally isomorphic to $\R^d=\R^{\sum_{i=1}^n d_i}$.
For example, technically, if $d_1=d_2=d_3=2$, then the elements of $\prod_{i=1}^n \R^{d_i}$ are of the form $\langle \langle x_1, x_2\rangle, \langle x_3, x_4\rangle, \langle x_5, x_6\rangle \rangle$, but this is trivially isomorphic to $\R^6$ by instead considering the element as $\langle x_1, x_2, x_3, x_4, x_5, x_6\rangle$.

Also observe (shown below) that if the original partitions were unit cube partitions, then the product partition is also a unit cube partition. 

\begin{fact}[Unit Cube Preservation]\label{unit-cube-preservation}
If $d_1,\ldots,d_n\in\N$ and $\P_1,\ldots,\P_n$ are {\em unit cube} partitions of $\R^{d_1},\ldots,\R^{d_n}$ respectively, then $\prod_{i=1}^n \P_i$ is also a unit cube partition.
\end{fact}
\begin{proof}
Each member of $\prod_{i=1}^n \P_i$ is of the form $\prod_{i=1}^{n}X_i$ where $X_i\in\P_i$. Since $\P_i$ is a unit cube partition, each $X_i$ is a product of translates of $[0,1)$, and thus $\prod_{i=1}^{n}X_i$ is also a product of translates of $[0,1)$, so the member is a unit cube.
\end{proof}

We can now present the main result of this section which is that if we take a product of partitions, and we have a guarantee for each $\P_i$ that it is $(k_i,\varepsilon_i)$-secluded, then we can guarantee the product partition is $(k,\varepsilon)$-secluded where $k$ is the product of the $k_i$'s and $\varepsilon$ is the minimum of the $\varepsilon_i$'s.

\begin{proposition}[Product Partition Seclusion Guarantees]\label{secluded-partition-product-proposition}
Let $n\in\N$. For each index $i\in[n]$, let $d_i,k_i\in\N$, $\varepsilon_i\in(0,\infty)$ and $\P_i$ be a $(k_i, \varepsilon_i)$-secluded partition of $\R^{d_i}$. Then the product partition $\P=\prod_{i=1}^{n}\P_i$ is a $(k,\varepsilon)$-secluded partition of $\R^d$ where $d=\sum_{i=1}^{n}d_i$, and $k=\prod_{i=1}^{n}k_i$, and $\varepsilon=\min_{i\in[n]}\varepsilon_i$.
\end{proposition}

\begin{proof}[Proof Sketch]
The basic idea is that for any point $\vec{p}\in\R^d$, we consider how many members of $\P$ intersect $\closure{B}_{\varepsilon}(\vec{p})$. 
Conceptually\footnote{In other words we identify the set $\R^d$ with $\R^{d_1}\times\R^{d_2}\times\cdots\times\R^{d_{n-1}}\times\R^{d_n}$}, we think of $\vec{p}$ as a sequence $\langle \vec{p}^{(i)} \rangle_{i=1}^{n}$ of $n$ points where the $i$th point $\vec{p}^{(i)}$ belongs to $\R^{d_i}$. 
Because we are working with the $\ell_\infty$ norm (that is the norm used by the definition of secluded), the $\varepsilon$ ball around $\vec{p}$ is the product of the $\varepsilon$ balls around each $\vec{p}^{(i)}$ which is smaller than the product of $\varepsilon_i$ balls around each $\vec{p}^{(i)}$ because we chose $\varepsilon$ as the minimum size. 
Thus, if the $\varepsilon$ ball around $\vec{p}$ intersects a member $X$ of the partition $\P$, then conceptually viewing $X$ as a sequence $\langle X_i \rangle_{i=1}^{n}$ where $X_i$ is a member of $\P_i$, it must be for each $i\in[n]$ that the $\varepsilon$ ball around $\vec{p}^{(i)}$ intersects $X_i$ (and thus so does the $\varepsilon_{i}$ ball since $\varepsilon_{i}\geq\varepsilon$). This means (for each $i\in[n]$) that $X_i$ is one of at most $k_i$ members of $\P_i$ because at most $k_i$ members of $\P_i$ intersect the $\varepsilon_i$ ball around $\vec{p}^{(i)}$ (by definition of $\P_i$ being $(k_i,\varepsilon_i)$-secluded). Thus $X$ is one of at most $\prod_{i=1}^{n}k_i=k$ members of $\P$. That is, there are at most $k$ members of $\P$ that intersect the $\varepsilon$ ball around $\vec{p}$ which is the definition of $\P$ being $(k,\varepsilon)$-secluded.
\end{proof}

Utilizing the construction above, we will now take a unit cube partition of \cite{geometry_of_rounding} for each $\R^{d_i}$ and take the product to obtain a new partition. Since the dimension of each $d_i$ is smaller than the dimension $d$, this allows us to get a larger value of $\varepsilon_i$ for each partition, and thus a larger value of $\varepsilon$ for the partition of $\R^d$ than if we had used one of the original partitions. The price we pay for this is that the value of $k$ also increases. The following result is nothing more than \Autoref{secluded-partition-product-proposition} where each partition in the product is specifically one of the partitions of \cite{geometry_of_rounding}. 

\restatableGlueConstruction*

\begin{proof}
Fix $d\in\N$. Let $d'=f(d)$ and $n=\ceil{\frac{d}{f(d)}}=\ceil{\frac{d}{d'}}$. Let $\P'$ be a $(d'+1, \frac{1}{2d'})$-secluded unit cube partition of $\R^{d'}$ (use the results of \cite{geometry_of_rounding}). 

By \Autoref{secluded-partition-product-proposition} and \Autoref{unit-cube-preservation}, $\P=\prod_{i=1}^n\P'$ is a $(k,\varepsilon)$-secluded unit cube partition of $\R^{n\cdot d'}$ where $k=(d'+1)^n$ and $\varepsilon=\frac{1}{2d'}$. Since $n\cdot d'=\ceil{\frac{d}{d'}}\cdot d'\geq d$, this trivially (by ignoring extra coordinates) gives a partition of $\R^d$ with these same properties (alternatively, see footnote\footnote{An alternate perspective is to let $d_1,\ldots,d_n$ be such that $\sum_{i=1}^n d_i=d$ and the first portion of the list $d_i=d'$, and the second portion of the list $d_i=d''\defeq d'-1$. Then let $\P'$ a $(d'+1,\frac{1}{2d'})$-secluded partition of $\R^{d'}$ as before, and let $\P''$ a $(d''+1,\frac{1}{2d''})$-secluded partition of $\R^{d''}$. Since $d''<d'$, $\P''$ is (by monotonicity) a $(d'+1,\frac{1}{2d'})$-secluded partition. Then take $\P_i=\P'$ when $d_i=d'$ and $\P_i=\P''$ when $d_i=d''$. Again, we get that $\P$ is $(k,\varepsilon)$-secluded for $k=(d'+1)^n$ and $\varepsilon=\frac{1}{2d'}$.}).
Recalling the definitions of $d'=f(d)$ and $n=\ceil{\frac{d}{f(d)}}$ gives the stated result.
\end{proof}

The above construction is very general. However, we can instantiate with various choices of parameters to get the following theorem. As discussed in the introduction, for these constructions the tolerance parameter $
\varepsilon(d)$ achieved is optimal up to a $O(\ln d)$ factor.  Below, $\weaksubexp(d)$ is $2^{o(d)}$.

\begin{restatable}{theorem}{restatableAchievableKEpsilonAllTogether}\label{achievable-k-epsilon-all-together}
Let $\varepsilon:\N\to(0,\infty)$. Then there exists $k:\N\to\N$ such that for every $d\in\N$ there exists a $(k(d),\varepsilon(d))$-secluded unit hypercube partition of $\R^d$, and $k$ has the following properties:
\begin{enumerate}
 \item\label{achieve-exp} If $\varepsilon(d)\in O(1)$, then $k(d)\in\eexp(d)$
   \item\label{achieve-weaksubexp} If $\varepsilon(d)\in o(1)$, then $k(d)\in\weaksubexp(d)$
  \item\label{achieve-poly} If $\varepsilon(d)\in O(\frac1d)$, then $k(d)\in\poly(d)$
 
\end{enumerate}
\end{restatable}

If all partitions $\P_1,\ldots,\P_n$ are ``efficiently computable'' in the sense that given an arbitrary point, $\vec{x}\in\R^{d_i}$ there is an efficient algorithm that computes a representation of the member of $\P_i$ containing $\vec{x}$, then the product partition is also ``efficiently computable'' because given some point $\vec{y}\in\R^d$, the member that it is contained in can be found by determining which member of $\P_1$ the point $\langle y_i \rangle_{i=1}^{d_1}$ is in, and independently determining which member of $\P_2$ the point $\langle y_i \rangle_{i=d_1+1}^{d_1+d_2}$ is in, etc. The member of $\prod_{i=1}^d \P_i$ that contains $\vec{y}$ is just the product of members.
This is an important property for using partitions as the basis of rounding schemes because an algorithm must determine which member/equivalence class a point is in (even if just implicitly). Because the partitions of \cite{geometry_of_rounding} are ``efficiently computable'' (see Proposition 11.2 in version 1), so are the partitions in this construction of \Autoref{basic-reclusive-gluing}.

%% file: 7_no_free_lunch.tex
\section{A No-Free-Lunch Theorem}\label{sec:NoFreeLunch}

\restatableNoFreeLunch*

\begin{proof}[Proof of \Autoref{no-free-lunch}]
Because $A$ is a deterministic algorithm mapping any point in $\R^d$ to $\R^d$ we can consider $A$ to be a mathematical function $A:\R^d\to\R^d$. Every mathematical function induces a natural partition of its domain which consists of the preimages/fibers of the function; that is 
\[
\P_A \defeq \set{A^{-1}(\vec{y}): \vec{y}\in\range(A)}.
\]
In other words, $\P_A$ is the partition defined by the equivalence relation on the domain $\R^d$ defined by $\vec{x}\sim \vec{x}\,'$ if and only if $A(\vec{x})=A(\vec{x}\,')$. Now we make a few claims about this partition.

\begin{subclaim}\label{no-free-lunch-subclaim-rounding-distance}
 For all $\vec{x}\in\R^d$, $\norm{A(\vec{x})-\vec{x}}\leq\epsilon$.
\end{subclaim}
\begin{subclaimproof}
Suppose for contradiction that there is some $\vec{x}\in\R^d$ such that $\norm{A(\vec{x})-\vec{x}}>\epsilon$. Let $\Lambda$ be some set (i.e. problem domain) and $f:\Lambda\to\R^d$ some function with $\vec{x}\in\range(f)$, and then let $\lambda\in\Lambda$ be some element which witnesses this (i.e. $f(\lambda)=\vec{x}$). Let $B$ be an $(\epsilon_0,\delta)$-approximation algorithm for $f$ (with respect to $\norm{\cdot}$) which has the property that on input $\lambda$, $B$ always returns $f(\lambda)=\vec{x}$ (i.e. $B$ approximates $f(\lambda)$ perfectly with probability $1$). Thus, on input $\lambda$, the algorithm $A\circ B$ always returns $A(\vec{x})$. But by hypothesis $\norm{A(\vec{x})-\vec{x}}>\epsilon$ which means that on input $\lambda$, $A\circ B$ always returns a value which is not an $\epsilon$-approximation to $f(\lambda)=\vec{x}$, which contradicts that $A\circ B$ is an $(\epsilon,\delta)$-approximation algorithm for $f$.
\end{subclaimproof}

If we were more careful we could actually get the bound above to $\epsilon-\epsilon'$, but we won't be that concerned. This allows us to show a bound on the diameter of all members of the partition $\P_A$.

\begin{subclaim}\label{no-free-lunch-subclaim-diameter}
Each member of $\P_A$  has diameter (with respect to $\norm{\cdot}$) at most $2\epsilon$.
\end{subclaim}
\begin{subclaimproof}
For any member $X\in\P_A$ we have by definition that for all $\vec{x},\vec{x}\,'\in A$ that $A(\vec{x})=A(\vec{x}\,')$. By the triangle inequality and \Autoref{no-free-lunch-subclaim-rounding-distance} we have
\begin{align*}
\norm{\vec{x}-\vec{x}\,'} 
&\leq \norm{\vec{x}-A(\vec{x})} + \norm{A(\vec{x})-A(\vec{x}\,')} + \norm{A(\vec{x}\,')-\vec{x}\,'}
\leq \epsilon + 0 + \epsilon
\end{align*}
which proves the claim.
\end{subclaimproof}

Then, by \Autoref{epsilon-diameter-bounded-any-norm} and \Autoref{no-free-lunch-subclaim-diameter}, there exists some point $\vec{p}\in\R^d$ such that $\ballcstd{\epsilon_0/2}{\vec{p}}$ intersects at least $\left(1+\frac{2(\epsilon_0/2)}{(2\epsilon)}\right)^d=\left(1+\frac{\epsilon_0}{2\epsilon}\right)^d$-many members of $\P_A$. Let this $\vec{p}$ be fixed for the remainder of the proof. We use this fact to put a lower bound on $k$.

\begin{subclaim}\label{no-free-lunch-subclaim-k-bound}
It holds that $k\geq(1-\delta)\cdot\left(1+\frac{\epsilon_0}{2\epsilon}\right)^d$.
\end{subclaim}
\begin{subclaimproof}
Let $T\subseteq\R^d$ be a set containing exactly one point in $X\cap\ballcstd{\epsilon_0/2}{\vec{p}}$ for each $X\in\P_A$ which intersects $\ballcstd{\epsilon_0/2}{\vec{p}}$. Because $\P_A$ is a partition, distinct $X,Y\in\P_A$ which intersect $\ballcstd{\epsilon_0/2}{\vec{p}}$, give distinct points regardless of the choice. Thus $\abs{T}\geq\left(1+\frac{\epsilon_0}{2\epsilon}\right)^d$.

Let $\Lambda$ be some set (i.e. problem domain) and $f:\Lambda\to\R^d$ some function with $\range(f)\cap T\not=\emptyset$, and then let $\lambda\in\Lambda$ be some element which witnesses this (i.e. $f(\lambda)\in T$). Let $B$ be an $(\epsilon_0,\delta)$-approximation algorithm for $f$ (with respect to $\norm{\cdot}$) which has the property that on input $\lambda$, $B$ returns a point selected uniformly\footnote{
    We discuss in a later footnote that the proof will still work even if perfectly uniform selection cannot be attained algorithmically.
} at random from $T$. This is a valid $(\epsilon_0,\delta)$-approximation because $f(\lambda)\in T$ and for all $\vec{x}\in T$, $\vec{x}\in\ballcstd{\epsilon_0/2}{\vec{p}}$ so by the triangle inequality $\norm{\vec{x}-f(\lambda)}\leq\epsilon_0$ which means $B$ always returns an $\epsilon_0$-estimate on input $\lambda$.

Because $A\circ B$ is a $(k,\delta)$-pseudodeterministic algorithm, there must be some set $S_\lambda\subseteq T$ with $\abs{S_\lambda}\leq k$ such that $\Pr[B(\lambda)\in S_\lambda]\geq 1-\delta$. Since $B(\lambda)$ is uniform over $T$, we have
\[
 1-\delta \leq \Pr[B(\lambda)\in S_\lambda] = \frac{\abs{S_\lambda}}{\abs{T}} \leq \frac{k}{\abs{T}}
\]
showing that
\[
k \geq (1-\delta)\cdot\abs{T} \geq (1-\delta)\cdot\left(1+\frac{\epsilon_0}{2\epsilon}\right)^d
\]
as claimed\footnote{
    As alluded to in a prior footnote, if perfectly uniform selection can't be achieved algorithmically, we instead can consider a sequence $B_1,B_2,B_3,\ldots$ of approximation algorithms for $f$ each defined the same way as $B$ but requiring only that $B_i$ distribute solutions close enough to uniformly that the probability of returning any of the $\abs{T}$ elements is at most $\frac{1}{\abs{T}}(1+\frac1i)$ so that $1-\delta \leq \Pr[B_i(\lambda)\in S_\lambda]  \leq \frac{\abs{S_\lambda}}{\abs{T}}(1+\frac1i) \leq \frac{k}{\abs{T}}(1+\frac1i)$. Since this is true for all $i\in\N$ the inequality passes through the limit and we get the same conclusion that $ 1-\delta \leq \frac{k}{\abs{T}}$.
}.
\end{subclaimproof}

Now in order to rearrange this lower bound on $k$ into a lower bound on $\epsilon$, we need to utilize an approximation\footnote{Specifically that $\ln(1+x)\leq\frac{x}2$ for small $x$.} which will require the assumption that $\epsilon\geq \epsilon_0$, so we claim and prove this next. This should not be surprising because it would be a fantastical result if there was a single deterministic algorithm $A$ which could improve the accuracy of every $\epsilon_0$-approximation algorithm to every function!

\begin{subclaim}\label{no-free-lunch-subclaim-epsilon-epsilon0}
We have that $\epsilon\geq \epsilon_0$.
\end{subclaim}
\begin{subclaimproof}
Let $\Lambda=\set{\lambda_-,\lambda_+}$ and let $\vec{v}\in\R^d$ be a $\norm{\cdot}$ unit vector. Let $f:\Lambda\to\R^d$ be defined by $f(\lambda_-)=-\epsilon_0\vec{v}$ and $f(\lambda_+)=\epsilon_0\vec{v}$. Let $B$ be the algorithm with always outputs $\vec{0}$ regardless of its input.

Then $B$ is an $(\epsilon_0,\delta)$-approximation algorithm for $f$ because $\vec{0}$ is an $\epsilon_0$-approximation for both $f(\lambda_-)$ and $f(\lambda_+)$. Because $A$ is deterministic\footnote{
    We could handle this portion even if $A$ was randomized using the fact that even a randomized $A$ would have to be $(k,\delta)$-pseudodeterministic if $A\circ B$ is, so the determinism is not essential here.
}, $A\circ B$ always returns the same value regardless of the input. Let $\vec{a}\in\R^d$ denote this value. Since $A\circ B$ is an $(\epsilon,\delta)$-approximation algorithm for $f$ it must be that $\norm{f(\lambda_-)-\vec{a}}\leq\epsilon$ and $\norm{f(\lambda_+)-\vec{a}}\leq\epsilon$, and because we have
\[
    2\epsilon_0 = \norm{f(\lambda_+)-f(\lambda_-)} \leq \norm{f(\lambda_+)-\vec{a}} + {f(\lambda_-)-\vec{a}}
\]
it must either be that $\norm{f(\lambda_-)-\vec{a}}\geq\epsilon_0$ or $\norm{f(\lambda_+)-\vec{a}}\geq\epsilon_0$. In either case, it shows $\epsilon_0\leq\epsilon$.
\end{subclaimproof}

Now we are ready to state the final inequality by taking the natural log of both sides of the inequality in \Autoref{no-free-lunch-subclaim-k-bound}. We then note by \Autoref{no-free-lunch-subclaim-epsilon-epsilon0} that $\frac{\epsilon_0}{2\epsilon}\leq\frac12$ and that for $x\leq\frac12$, $\ln(1+x)\geq\frac{x}2$. And lastly, because $\delta\in(0,\frac12]$, we have $\ln(1-\delta)\geq\ln(\frac12)$.
\begin{align*}
    \ln(k) &\geq \ln(1-\delta) + d\ln\left(1+\frac{\epsilon_0}{2\epsilon}\right)\\
           &\geq \ln(1-\delta) + d\cdot  \frac{\epsilon_0}{4\epsilon}\\
           &\geq \ln(\tfrac12) + d\cdot  \frac{\epsilon_0}{4\epsilon}.
\end{align*}
Solving for $\epsilon$ we get
\[
\epsilon\geq\epsilon_0\cdot\frac{d}{4\ln(2k)}
\]
as desired which completes the proof.
\end{proof}

In the case of the $\ell_\infty$ norm, the bounds of \Autoref{no-free-lunch} can be nearly matched (up constants and logarithmic factors) in the regime of interest where the pseudodeterminism/replicability value $k$ is polynomial in the spacial dimension $d$. This is shown in the next result which says there is a deterministic function/algorithm which does everything described in \Autoref{no-free-lunch} with $\epsilon=2d\cdot\epsilon_0$. The constant $2$ here can be replaced by any constant if one is willing to increase $k$ from the value $d+1$ stated in the result below to some greater polynomial. 

\input{universal_rounding_algorithm_ell_infty}

\begin{proof}[Proof Sketch]
This follows by using a scaled $(d+1,\frac1{2d})$-secluded partition with unit diameter members as a deterministic rounding scheme. A $(\poly(d),O(\frac1d))$-secluded partition with unit diameter members as in \Autoref{achievable-k-epsilon-all-together} can also be used to trade off polynomial factors in the first  parameter (degree) with constant factors in the second (tolerance).
\end{proof}

%% file: universal_rounding_algorithm_ell_infty.tex
\begin{theorem}\label{universal_rounding_function}
Let $d\in\N$ and $\epsilon_0\in(0,\infty)$. Let $\epsilon=\epsilon_0\cdot 2d$. There is an efficiently computable function/algorithm $A_{\epsilon}:\R^d\to\R^d$ with the following two properties: 
\begin{enumerate}
    \item For any $x\in\R^d$ and any $\hat{x}\in\ballcinf{\epsilon_0}{x}$ it holds that $A_{\epsilon}(\hat{x})\in\ballcinf{\epsilon}{x}$.
    \item For any $x\in\R^d$ the set $\set{A_{\epsilon}(\hat{x})\colon \hat{x}\in\ballcinf{\epsilon_0}{x}}$ has cardinality at most $d+1$.
\end{enumerate}
Informally, these two conditions are (1) if $\hat{x}$ is an $\epsilon_0$-approximation of $x$ (with respect to $\ell_\infty$), then $A_{\epsilon}(\hat{x})$ is an $\epsilon$ approximation of $x$, and (2) $A_{\epsilon}$ maps every $\epsilon_0$ approximation of $x$ to one of at most $d+1$ possible values.
\end{theorem}

%% file: appendix_measure_theory.tex
\section{Measure Theory}
\label{sec:measure-theory}
\renewcommand{\A}{\mathcal{A}}
\newcommand{\F}{\mathcal{F}}

Throughout this section, we use the word ``countable'' to mean finite or countably infinite.

\begin{fact}[$\ell_\infty$ Diameter Ball]\label{diameter-ball}
Let $d\in\N$ and $X\subseteq\R^d$ be a bounded set with diameter $D$ (with respect to $\ell_\infty$). Then there exists $\vec{p}\in\R^d$ such that $X\subseteq \closure{B}_{D/2}(\vec{p})$. As a consequence, $m_{out}(X)\leq D^d$ where $m_{out}$ denotes outer Lebesgue measure.
\end{fact}
\begin{proof}[Proof Sketch]
For each coordinate $i\in[d]$, consider the set $X_i=\set{\pi_i(\vec{x})\colon \vec{x}\in X}\subseteq\R$ of the $i$th coordinates of each point in $X$. The infimum and supremum are distance at most $D$ apart, because otherwise there would be points $\vec{y},\vec{z}\in X$ such that $\abs{\pi_i(\vec{z})-\pi_i(\vec{y})}>D$ which means $\norm{\vec{z}-\vec{y}}_\infty>D$. Thus taking $\vec{p}=\langle \frac{\inf(X_i)+\sup(X_i)}{2} \rangle_{i=1}^d$ we have $X\subseteq\prod_{i=1}^d[\inf(X_i),\sup(X_i)]\subseteq\vec{p}+[-\frac{D}{2}, \frac{D}{2}]^d=\ballc{D/2}{\vec{p}}$.
\end{proof}

\begin{fact}\label{disjoint-uncountable}
  If $\mu$ is a measure and $\A$ is a (possibly uncountable) family of pairwise disjoint measurable sets, then
  \[
    \mu(\bigsqcup_{A\in\A}A)\geq\sum_{A\in\A}\mu(A).
  \]
\end{fact}
\begin{proof}
  By definition of the arbitrary summation (c.f. \cite[p.~11]{folland_real_1999}) we have
  \[
    \sum_{A\in\A}\mu(A)\defeq\sup\set{\sum_{A\in\F}\mu(A):\F\subseteq\A,\;\F\text{ finite}}
  \]
  and for any $\F\subseteq\A$ we have
  \[
    \mu(\bigsqcup_{A\in\A}A)\geq \mu(\bigsqcup_{A\in\F}A)=\sum_{A\in\F}\mu(A).
  \]
  Thus $\mu(\bigsqcup_{A\in\A}A)$ is an upper bound for the set $\set{\sum_{A\in\F}\mu(A):\F\subseteq\A,\;\F\text{ finite}}$ and thus greater than or equal to the supremum.
\end{proof}

\begin{fact}[Interchange of Countable Sums with Non-negative Terms]\label{interchange-of-countable-sums}
If $I,J$ are countable sets, and $a_{i,j}\geq0$ for all $(i,j)\in I\times J$, then
\[
    \sum_{i\in I}\sum_{j\in J}a_{i,j} = \sum_{j\in J}\sum_{i\in I}a_{i,j}
\]
\end{fact}
\begin{proof}
This can be proved directly via basic analysis methods if $I$ and $J$ are assumed to be $\N$ and the definition of the infinite sum as a limit of finite sums is used. Alternatively, viewing the summation as an integral over a countable measure space, this can be viewed as a corollary to Tonelli's theorem.
\end{proof}

\renewcommand{\A}{\mathcal{A}}
\newcommand{\FF}{\mathcal{F}}

\begin{lemma}[Exact Measure of Multiplicity]\label{exact-measure-of-multiplicity}
Let $n\in\N$. Let $X$ be a measurable set in some measure space (the measure being denoted by $\mu$) and $\A$ a countable family of measurable subsets of $X$ such that for each $x\in X$, $x$ belongs to {\em exactly} $n$ members of $\A$. Then
\[
\sum_{A\in\A}\mu(A) = n\cdot\mu(X).
\]
\end{lemma}
\begin{proof}
We note that if $n=0$, then the statement is trivially true because $\A$ is either empty or contains just the empty set, and in either case $\sum_{A\in\A}\mu(A) = 0 = 0\cdot\mu(X)$ if we use the standard convention that the empty sum is $0$.

For any $\FF\subseteq\A$, let $G_{\FF}=\bigcap_{A\in\FF}A$ noting that this is a countable intersection of measurable sets, so it is measurable (mnemonically, the $G$ represents an intersection as it does in the notation for $G_{\delta}$ sets).

Let $\A\choose n$ denote all subsets of $\A$ of size $n$ noting that because $\A$ is countable, so is $\A\choose n$. Observe that for distinct $\FF,\FF'\in{\A\choose n}$, the sets $G_{\FF}$ and $G_{\FF'}$ are disjoint because
\[
    G_{\FF}\cap G_{\FF'} = \left(\bigcap_{A\in\FF}A\right) \cap \left(\bigcap_{A\in\FF'}A\right) = \bigcap_{A\in\FF\cup\FF'}A
\]
and since $\FF$ and $\FF'$ are distinct and each contain $n$ items, $\abs{\FF\cup\FF'}\geq n+1$, and by assumption no point in $X$ belongs to $n+1$ members, so $\bigcap_{A\in\FF\cup\FF'}A=\emptyset$. Furthermore, for each $x\in X$, since $x$ belongs to exactly $n$ members $A_1,\ldots,A_n$ of $\A$, taking $\FF=\set{A_1,\ldots,A_n}$ we have $x\in G_{\FF}$ which shows that $\set{G_{\FF}\colon \FF\in{\A\choose n}}$ is a partition of $X$ into countably many measurable sets (allowing that some $G_{\FF}$ might be empty).

The last observation we need is that for any $\FF\in{\A\choose n}$ and any $A\in\A$, it holds that if $A\in\FF$, then $A\supseteq G_{\FF}$ and if $A\not\in\FF$ then $A\cap G_{\FF}=\emptyset$. To see this, note that for any $x\in G_{\FF}$, $x$ belongs to each of the $n$ members of $\FF\subseteq\A$, and since by assumption $x$ belongs to exactly $n$ members of $\A$, it does not belong to any other members of $\A$.

Now we have the following chain of equalities:
\begin{align*}
    \sum_{A\in\A}\mu(A) &= \sum_{A\in\A}\mu(A\cap X) \tag{$A\subseteq X$ so $A\cap X=A$} \\
    &= \sum_{A\in\A}\mu\left(A\cap\left[\bigsqcup_{\FF\in{\A\choose n}}G_{\FF}\right]\right) \tag{Set equality; the $G_{\FF}$ partition $X$} \\
    %
    &= \sum_{A\in\A}\mu\left(\bigsqcup_{\FF\in{\A\choose n}}\left[A\cap G_{\FF}\right]\right) \tag{Set equality} \\
    %
    &= \sum_{A\in\A}\left[\sum_{\FF\in{\A\choose n}}\mu\left(A\cap G_{\FF}\right)\right] \tag{Countable additivity of measures} \\
    &= \sum_{\FF\in{\A\choose n}}\left[\sum_{A\in\A}\mu\left(A\cap G_{\FF}\right)\right] \tag{Interchange sums by \Autoref{interchange-of-countable-sums}} \\
    &= \sum_{\FF\in{\A\choose n}}\left[\sum_{A\in\A}
    \begin{cases}
        \mu\left(A\cap G_{\FF}\right) = \mu\left(G_{\FF}\right) & A\in\FF \\
        \mu\left(A\cap G_{\FF}\right) = \mu(\emptyset) = 0 & A\not\in\FF
    \end{cases}
    \right] \tag{Previous paragraph} \\
    &= \sum_{\FF\in{\A\choose n}}\left[\sum_{A\in\FF}\mu\left(G_{\FF}\right)\right] \tag{Remove $0$ terms from summation} \\
    &= \sum_{\FF\in{\A\choose n}}\left[n\cdot\mu\left(G_{\FF}\right)\right] \tag{$\abs{\FF}=n$} \\
    &= n\sum_{\FF\in{\A\choose n}}\mu\left(G_{\FF}\right) \tag{Linearity of summation} \\
    &= n\cdot\mu\left(\bigsqcup_{\FF\in{\A\choose n}}G_{\FF}\right) \tag{Countable additivity of measures} \\
    &= n\cdot\mu\left(X\right) \tag{Set equality; the $G_{\FF}$ partition $X$} \\
\end{align*}
This proves the result.
\end{proof}

\begin{lemma}[Upper Bound Measure of Multiplicity]\label{upper-bound-measure-of-multiplicity}
Let $n\in\N$. Let $X$ be a measurable set in some measure space (the measure being denoted by $\mu$) and $\A$ a countable family of measurable subsets of $X$ such that for each $x\in X$, $x$ belongs to {\em at most} $n$ members of $\A$. Then
\[
\sum_{A\in\A}\mu(A) \leq n\cdot\mu(X).
\]
\end{lemma}
\begin{proof}
As in the last proof, for any $\FF\subseteq\A$, let $G_{\FF}=\bigcap_{A\in\FF}A$ noting that this is a countable intersection of measurable sets, so it is measurable (mnemonically, the $G$ represents an intersection as it does in the notation for $G_{\delta}$ sets). And for any $k\in[n]\cup\set{0}$, let $\A\choose k$ denote all subsets of $\A$ of size $k$ noting that because $\A$ is countable, so is $\A\choose k$.

For each $k\in[n]\cup\set{0}$, let
\begin{align*}
    S_k &= \set{x\in X\colon \text{$x$ belongs to {\em exactly} $k$ members of $\A$}} \\
    S_k' &= \set{x\in X\colon \text{$x$ belongs to {\em at least} $k$ members of $\A$}}
\end{align*}
We will show that $S_k$ and $S_k'$ are measurable.

To show that the $S_k'$ are measurable, note that for any $k\in[n]\cup\set{0}$, $S_k'$ can be expressed as $S_k' = \bigcup_{\FF\in{\A\choose k}}G_{\FF}$. This is because for any $x\in X$, if $x$ belongs to at least $k$ members of $\A$, then there is a subset $\FF\subseteq\A$ with $\abs{\FF}=k$ such that $x\in\bigcap_{A\in\FF}A=G_{\FF}$. Conversely, if $x\in\bigcap_{A\in\FF}A=G_{\FF}$, then there is some $\FF\in{\A\choose k}$ (i.e. some $\FF\subseteq\A$ with $\abs{\FF}=k$) such that $x\in G_{\FF}=\bigcap_{A\in\FF}A$, so $x$ belongs to at least $k$ members of $\A$. Thus, since $\A$ is countable, so is $\A\choose k$ (for each $k$) implying that each $S_k'$ is a countable union of measurable sets, so is itself measurable.

To show the measurability of each $S_k$, first consider $k=n$. Observe that $S_n=S_n'$ because by assumption each $x\in X$ belongs to at most $n$ members of $\A$, so it belongs to exactly $n$ members if and only if it belongs to at least $n$ members. Thus, $S_n$ is also measurable.

Now for $k\in[n-1]\cup\set{0}$ observe that $S_k=S_k'\setminus S_{k+1}'$ because some $x\in X$ belongs to exactly $k$ members of $\A$ if and only if it belongs to at least $k$ members and does not belong to at least $k+1$ members of $\A$. Thus, for $k\in[n-1]\cup\set{0}$, $S_k$ is the set difference of two measurable sets, so is itself measurable.

Finally, note that $\set{S_k\colon k\in[n]\cup\set{0}}$ is a partition of $X$ (allowing the possibility that some $S_k$ are empty) because every point of $x$ belongs to some number of members of $\A$, and that number is (by assumption) between $0$ and $n$ inclusive.

Now we have the following chain of inequalities:
\begin{align*}
    \sum_{A\in\A}\mu(A) &= \sum_{A\in\A}\mu(A\cap X) \tag{$A\subseteq X$ so $A\cap X=A$} \\
    &= \sum_{A\in\A}\mu\left(A\cap\left[\bigsqcup_{k\in[n]\cup\set{0}}S_k\right]\right) \tag{Set equality; the $S_k$ partition $X$} \\
    &= \sum_{A\in\A}\mu\left(\bigsqcup_{k\in[n]\cup\set{0}}\left[A\cap S_k\right]\right) \tag{Set equality} \\
    &= \sum_{A\in\A}\left[\sum_{k\in[n]\cup\set{0}}\mu\left(A\cap S_k\right)\right] \tag{Countable additivity of measures} \\
    &= \sum_{k\in[n]\cup\set{0}}\left[\sum_{A\in\A}\mu\left(A\cap S_k\right)\right] \tag{Interchange sums by \Autoref{interchange-of-countable-sums}} \\
    &= \sum_{k\in[n]\cup\set{0}}\left[k\cdot\mu\left(S_k\right)\right] \tag{By \Autoref{exact-measure-of-multiplicity}; see details below} \\
    &= \sum_{k\in[n]}\left[k\cdot\mu\left(S_k\right)\right] \tag{$k=0$ term is $0$} \\
    &\leq \sum_{k\in[n]}\left[n\cdot\mu\left(S_k\right)\right] \tag{$k\leq n$} \\
    &= n\sum_{k\in[n]}\left[\mu\left(S_k\right)\right] \tag{Linearity of summation} \\
    &= n\cdot\mu\left(\bigsqcup_{k\in[n]}S_k\right) \tag{Countable additivity of measures} \\
    &\leq n\cdot\mu\left(X\right) \tag{Set inequality; the $S_k$ partition $X$, but $S_0$ is missing from the union}
\end{align*}

After justifying the use of \Autoref{exact-measure-of-multiplicity}, this completes the proof. For each $k\in[n]\cup\set{0}$, let $X_k=S_k$ and $\A_k=\set{A\cap S_k\colon A\in\A}$. Then observe that for each $x\in X_k=S_k$, by the definition of $S_k$, $x$ belongs to exactly $k$ members of $\A$, and thus (since it also belongs to $S_k$) belongs to exactly $k$ members of $\A_k$. Applying \Autoref{exact-measure-of-multiplicity} once for each $k$ with $X=X_k$ and $\A=\A_k$ shows that 
\[
 \sum_{A\in\A}\mu(A\cap S_k) = \sum_{A'\in \A_k}\mu(A') = k\cdot\mu(X_k) = k\cdot\mu(S_k)
\]
(the middle equality is where \Autoref{exact-measure-of-multiplicity} was applied). This is what we claimed in the long chain of equalities above and completes the proof.

\end{proof}

\begin{restatable}[Lower Bound Cover Number]{corollary}{restatableLowerBoundCoverNumber}\label{lower-bound-cover-number}
\renewcommand{\A}{\mathcal{A}}
Let $X$ be a measurable set in some measure space (the measure being denoted by $\mu$) such that $0<\mu(X)<\infty$. Let $\A$ be a countable family of measurable subsets of $X$ such that $\sum_{A\in\A}\mu(A)<\infty$. Then there exists $x\in X$ such that $x$ belongs to at least $\ceil*{\frac{\sum_{A\in\A}\mu(A)}{\mu(X)}}$-many members of $\A$.
\end{restatable}
\begin{proof}

First observe that by hypothesis, $\ceil{\frac{\sum_{A\in\A}\mu(A)}{\mu(X)}}$ is finite. Suppose for contradiction that each $x\in X$ belongs to strictly less than $\ceil{\frac{\sum_{A\in\A}\mu(A)}{\mu(X)}}$-many members of $\A$. Let $n=\ceil{\frac{\sum_{A\in\A}\mu(A)}{\mu(X)}}-1$ (noting that $n<\frac{\sum_{A\in\A}\mu(A)}{\mu(X)}$). Then each $x\in X$ belongs to at most $n$-many members of $A$, so we have
\begin{align*}
    \sum_{A\in\A}\mu(A) &\leq n\cdot\mu(X) \tag{\Autoref{upper-bound-measure-of-multiplicity}} \\
    &< \frac{\sum_{A\in\A}\mu(A)}{\mu(X)}\mu(X) \tag{$0<\mu(X)<\infty$ and $n<\frac{\sum_{A\in\A}\mu(A)}{\mu(X)}$} \\
    &= \sum_{A\in\A}\mu(A)
\end{align*}
which is a contradiction.
\end{proof}

\begin{remark}\label{infinite-cover-case}
In \Autoref{lower-bound-cover-number} above, it was important that we required $\sum_{A\in\A}\mu(A)$ to be finite. If we allowed it to be infinite, then the claim would have been that there was some $x\in X$ belonging to infinitely many members of $\mathcal{A}$, but this is in general not true (see \Autoref{harmonic-cover-example} below). Nonetheless, it is true (and a straightforward corollary of the above) that if $\sum_{A\in\A}\mu(A)=\infty$, then for any $n\in\N_0$, there exists a point $x_n\in X$ that is contained in at least $n$-many sets of $\A$. The distinction is that this point might have to depend on the choice of $n$.
\end{remark}

\begin{example}[Harmonic Cover of Open Unit Interval]\label{harmonic-cover-example}
Let $X=(0,1)$ be equipped with the Borel or Lebesgue measure $\mu$. Let $\mathcal{A}=\set{(0,\tfrac1i): i\in\N}$. Then $\sum_{A\in\A}\mu(A)=\sum_{i\in\N}\tfrac1i=\infty$. For any $n\in\N$, we can consider the point $x_n=\tfrac1{n+1}$ which is contained in $(0,\tfrac1i)$ for $i\in[n]$ and not for any other $i$, so it belongs to exactly $n$ sets in $\A$.

However, no point in $X$ belongs to infinitely many sets in $\A$. To see this, consider an arbitrary point $x\in X=(0,1)$. Then for sufficiently large $i\in\N$, $x\not\in(0,\tfrac1i)$ so $x$ belongs to only finitely many members of $\A$.
\end{example}

The prior three results have been stated in typical measure theory notation, but in the body of the paper we present \Autoref{lower-bound-cover-number} as follows for $\R^d$ specifically with notation matching what is used elsewhere in the paper.

\hypertarget{hypertarget-lower-bound-cover-number-Rd}{}

\restatableLowerBoundCoverNumberRd
\begin{proof}
This follows trivially from \Autoref{lower-bound-cover-number} and \Autoref{infinite-cover-case}.
\end{proof}

\hypertarget{hypertarget-iso-diametric-ineq}{}

\restatableIsodiametricIneq*
The outline of the proof below was found by the authors in \cite{noauthor_answer_2014} who cites \cite{gruber_convex_2007}.
\begin{proof}
Because $D=\diamstd(A)$, also $\diamstd(\clos{A})=D$ as well. Consider the set $\clos{A}-\clos{A}\defeq\set{\vec{a}'-\vec{a}'':\vec{a}',\vec{a}''\in\clos{A}}$ noting that $\clos{A}-\clos{A}$ is closed (and thus Borel measurable), centrally symmetric\footnote{
    That is, for any $\vec{a}\in\clos{A}-\clos{A}$, also $-\vec{a}\in\clos{A}-\clos{A}$, because $\vec{a}=\vec{a}'-\vec{a}''$ for some $\vec{a}',\vec{a}''\in\clos{A}$, so also $\clos{A}-\clos{A}\ni \vec{a}''-\vec{a}'=-\vec{a}$.
}, and has diameter\footnote{
    Given any two vectors $\vec{a},\vec{b}\in\clos{A}-\clos{A}$ we have $\vec{a}=\vec{a}'-\vec{a}''$ and $\vec{b}=\vec{b}'-\vec{b}''$ for some $\vec{a}',\vec{a}'',\vec{b}',\vec{b}\in\clos{A}$. So $\norm{\vec{a}-\vec{b}} = \norm{(\vec{a}'-\vec{a}'')-(\vec{b}'-\vec{b}'')} = \norm{(\vec{a}'-\vec{b}')+(\vec{b}''-\vec{a}'')} \leq \norm{(\vec{a}'-\vec{b}')}+\norm{(\vec{b}''-\vec{a}'')} \leq 2D$, so the diameter of $\clos{A}-\clos{A}$ is at most $2D$.
} at most $2D$. This implies\footnote{
    If there was some $\vec{a}\in\clos{A}-\clos{A}$ with $\norm{\vec{a}}>D$, then by central symmetry, also $-\vec{a}\in\clos{A}-\clos{A}$, so $2D=\diam{\clos{A}-\clos{A}}\geq \norm{\vec{a}-(-\vec{a})} = 2\norm{\vec{a}} > 2D$ which would be a contradiction.
} that $\clos{A}-\clos{A}\subseteq\ballcstd{D}{\vec{0}}$. Thus, $m(\clos{A}-\clos{A})\leq m(\ballc{D}{\vec{0}})$.

Also, letting $-\clos{A}\defeq\set{-\vec{a}\colon \vec{a}\in\clos{A}}$ we have that $\clos{A}-\clos{A}$ is the Minkowski sum $\clos{A}-\clos{A}=\clos{A}+(-\clos{A})$. This allows us to use the \namefullref{generalized-brunn-minkowski-inequality} to obtain
\begin{align*}
    m(\clos{A}-\clos{A}) &=  m(\clos{A}+(-\clos{A})) \\
    &\geq \left[m(\clos{A})^\frac1d + m(-\clos{A})^\frac1d \right]^d \tag{Brunn-Minkowski}\\
    &= \left[2\cdot m(\clos{A})^\frac1d\right]^d \\
    &= 2^d \cdot m(\clos{A})
\end{align*}
Combining this with the inequality at the end of the last paragraph gives $m(\clos{A})\leq \frac{1}{2^d}\cdot m(\ballc{D}{\vec{0}})$.

We complete the proof noting a few simple inequalities. First, since $A\subseteq\clos{A}$, we have $m_{out}(A)\leq m(\clos{A})$. Second, by the scaling of Lebesgue measure, $m(\ballc{D}{\vec{0}})=m(D\cdot\ballc{1}{\vec{0}})=D^d\cdot m(\ballc{1}{\vec{0}})$. Third, $m(\ballo{1}{\vec{0}})=m(\ballc{1}{\vec{0}})$ because for any $\varepsilon\in(0,\infty)$ we have
\[
m(\ballo{1}{\vec{0}}) \leq m(\ballc{1}{\vec{0}}) \leq m(\ballo{1+\varepsilon}{\vec{0}}) = (1+\varepsilon)^d\cdot m(\ballo{1}{\vec{0}}).
\]
Combining all of this gives the result:
\[
m_{out}(A) \leq m(\clos{A}) \leq \left(\frac D 2\right)^d\cdot m(\ballc{1}{\vec{0}}) = \left(\frac D 2\right)^d\cdot m(\ballo{1}{\vec{0}}).
\]

\end{proof}

%% file: appendix_minkowski_sums.tex
\section{Minkowski Sums}
\label{sec:minkowski-sums}

\restatableMinkowskiBubbleUnion*
\begin{proof}
We show this only for the open balls. Switching all open balls in the proof with closed ones gives the proof for closed balls.

($\subseteq$) A generic element of $X+\ballo{\varepsilon}{\vec{0}}$ is $\vec{x}+\vec{b}$ for some $\vec{x}\in X$ and $\vec{b}\in \ballo{\varepsilon}{\vec{0}}$ which means $\vec{x}+\vec{b}\in\vec{x}+\ballo{\varepsilon}{\vec{0}}=\ballo{\varepsilon}{\vec{x}}\subseteq \bigcup_{\vec{x}\in X}\ballo{\varepsilon}{\vec{x}}$.

($\supseteq$) Given $\vec{y}\in \bigcup_{\vec{x}\in X}\ballo{\varepsilon}{\vec{x}}$ there is some particular $\vec{x}\in X$ such that $y\in\ballo{\varepsilon}{\vec{x}}=\vec{x}+\ballo{\varepsilon}{\vec{0}}$ which means $\vec{y}=\vec{x}+\vec{b}$ for some $\vec{b}\in \ballo{\varepsilon}{\vec{0}}$, and since $\vec{x}\in X$, we have $\vec{y}\in X+\ballo{\varepsilon}{\vec{0}}$.
\end{proof}

\restatableSumOfBalls*
\begin{proof}
($\subseteq$) A generic element of $\ballo{\alpha}{\vec{0}}+\ballo{\beta}{\vec{0}}$ is $\vec{a}+\vec{b}$ for $\vec{a}\in\ballo{\alpha}{\vec{0}}$ and $\vec{b}\in\ballo{\beta}{\vec{0}}$. Then $\norm{\vec{a}}<\alpha$ and $\norm{\vec{b}}<\beta$, so $\norm{\vec{a}+\vec{b}}<\alpha+\beta$ showing $\vec{a}+\vec{b}\in\ballo{\alpha+\beta}{\vec{0}}$.

($\supseteq$) Let $\vec{x}\in\ballo{\alpha+\beta}{\vec{0}}$ which implies $\norm{\vec{x}}<\alpha+\beta$. If $\norm{\vec{x}}<\alpha$, then $\vec{x}\in\ballo{\alpha}{\vec{0}}$ and $\vec{0}\in\ballo{\beta}{\vec{0}}$, so $\vec{x}=\vec{x}+\vec{0}\in\ballo{\alpha}{\vec{0}}+\ballo{\beta}{\vec{0}}$. Similarly, if $\norm{\vec{x}}<\beta$, then $\vec{x}=\vec{0}+\vec{x}\in\ballo{\alpha}{\vec{0}}+\ballo{\beta}{\vec{0}}$. In either case we would be done, so we may now assume that $\norm{\vec{x}}\geq\alpha,\beta$. Let $\varepsilon=\alpha+\beta-\norm{\vec{x}}\in(0,\infty)$. Since $\norm{\vec{x}}\geq\alpha$, we have $\varepsilon\leq\beta$, and because $\norm{\vec{x}}\geq\beta$, we have $\varepsilon\leq\alpha$. This shows $\frac\varepsilon2<\alpha,\beta$. Let $\vec{a}=(\alpha-\frac\varepsilon2)\frac{\vec{x}}{\norm{\vec{x}}}$ and $\vec{b}=(\beta-\frac\varepsilon2)\frac{\vec{x}}{\norm{\vec{x}}}$ noting that $\norm{\vec{a}}=\alpha-\frac\varepsilon2\in(0,\alpha)$ and $\norm{\vec{b}}=\beta-\frac\varepsilon2\in(0,\beta)$. Also, note that $\vec{a}+\vec{b}=(\alpha-\frac\varepsilon2+\beta-\frac\varepsilon2)\frac{\vec{x}}{\norm{\vec{x}}}=\norm{\vec{x}}\frac{\vec{x}}{\norm{\vec{x}}}=\vec{x}$ which shows $\vec{x}\in\ballo{\alpha}{\vec{0}}+\ballo{\beta}{\vec{0}}$.
\end{proof}

\restatableAntisymmetryContainmentPreClaimMinkowskiSum*
\begin{proof}
We show this only for the closed balls. Switching all closed balls in the proof with open ones gives the proof for open balls.

($\Longrightarrow$) If $\ballc{\varepsilon}{\vec{p}}\cap X\not=\emptyset$, then there exists $\vec{y}\in\ballc{\varepsilon}{\vec{p}}\cap X$, and since $\vec{y}\in\ballc{\varepsilon}{\vec{p}}$ we have $\norm{\vec{y}-\vec{p}}\leq\varepsilon$ so $\vec{p}\in\ballc{\varepsilon}{\vec{y}}=\vec{y}+\ballc{\varepsilon}{\vec{0}}$, and since $\vec{y}\in X$, $\vec{y}+\ballc{\varepsilon}{\vec{0}}\subseteq X+\ballc{\varepsilon}{\vec{0}}$ showing that $\vec{p}\in X+\ballc{\varepsilon}{\vec{0}}$.

($\Longleftarrow$) If $\vec{p}\in X+\ballc{\varepsilon}{\vec{0}}$ then there exists $\vec{x}\in X$ and $\vec{b}\in\ballc{\varepsilon}{\vec{0}}$ such that $\vec{p}=\vec{x}+\vec{b}$. Thus $\vec{p}-\vec{x}=\vec{b}$, so $\norm{\vec{p}-\vec{x}}=\norm{\vec{b}}\leq\varepsilon$, so $\vec{x}\in\ballc{\varepsilon}{\vec{p}}$. Since $\vec{x}$ belongs to both $X$ and $\ballc{\varepsilon}{\vec{p}}$, their intersection is non-empty.
\end{proof}

%% file: appendix_other.tex
\section{Additional Facts}
\label{sec:other}

\hypertarget{hypertarget-smaller-ceiling}{}

\restatableSmallerCeiling*
\begin{proof}
Let $n=\ceil{\alpha}$. This implies that $\alpha>n-1$ (otherwise $\alpha\leq n-1$ so $\ceil{\alpha}\leq n-1$). Thus $(n-1,\alpha)$ is non-empty and we can take any $\gamma\in(n-1,\alpha)$. Then $n-1<\ceil{\gamma}\leq\ceil{\alpha}=n$ showing $\ceil{\gamma}=n$ as well.
\end{proof}

\hypertarget{hypertarget-brunn-minkowski-bound-lemma}{}

\restatableBrunnMinkowskiBoundLemma*
\begin{proof}
We will show that $(x^{1/d}+\alpha)^d - x(1+\alpha)\geq 0$ for these parameters. Let For $d$, $\alpha$ as above, let $f:[0,1]\to\R$ be defined by $f(x)=(x^{1/d}+\alpha)^d - x(1+\alpha)$. Observe that $f(0)=\alpha^2\geq0$ and $f(1)=(1+\alpha)^d-(1+\alpha)^d=0$. We will now prove that $f$ is convex on the domain\footnote{Actually we could have defined the domain of $f$ to be $[0,\infty)$ and we show that $f$ is convex on that domain. However, we only have need of the interval $[0,1]$.} $[0,1]$. This will be sufficient to prove the claim because $f$ is also non-negative at $0$ and at $1$.

We show that $f$ is convex on $[0,1]$ by considering its second derivative on $(0,1]$. 
\begin{align*}
    \frac{d}{dx}f(x) & = \frac{d}{dx} \left[(x^{1/d}+\alpha)^d - x(1+\alpha)^d\right] \\
    &= d(x^{1/d}+\alpha)^{d-1}\cdot\frac{1}{d}x^{1/d-1} - (1+\alpha)^d \\
    &= (x^{1/d}+\alpha)^{d-1}x^{1/d-1} - (1+\alpha)^d \\
\end{align*}
where we use the convention that $0^0=1$. Then
\begin{align*}
    \frac{d^2}{dx^2}f(x) & = \frac{d}{dx} \left[(x^{1/d}+\alpha)^{d-1}x^{1/d-1} - (1+\alpha)^d\right] \\
    &= (x^{1/d}+\alpha)^{d-1}\cdot\left(\frac1d-1\right)x^{1/d-2}+x^{1/d-1}\cdot(d-1)(x^{1/d}+\alpha)^{d-2}\frac1d x^{1/d-1} \\
    &= (x^{1/d}+\alpha)(x^{1/d}+\alpha)^{d-2}\left(-\frac{d-1}d\right)\left(x^{1/d-2}\right) + (x)\left(x^{1/d-2}\right)\left(\frac{d-1}{d}\right)\left(x^{1/d}+\alpha\right)^{d-2}\left(x^{1/d-1}\right) \\
    &= \frac{(x^{1/d}+\alpha)^{d-2}(d-1)\left(x^{1/d-2}\right)}{d}\left[ -(x^{1/d}+\alpha) + (x)\left(x^{1/d-1}\right) \right] \\
    &= \frac{(x^{1/d}+\alpha)^{d-2}(d-1)\left(x^{1/d-2}\right)}{d}\left[ -(x^{1/d}+\alpha) + \left(x^{1/d}\right) \right] \\
    &= \frac{-\alpha(x^{1/d}+\alpha)^{d-2}(d-1)\left(x^{1/d-2}\right)}{d} \\
    &= \frac{-\alpha(x^{1/d}+\alpha)^{d}(d-1)\left(x^{1/d}\right)}{d(x^{1/d}+\alpha)^{2}x^2}.
\end{align*}
Note that $\frac{d^2}{dx^2}f(x)\leq0$ for $x\in(0,1]$ since $\alpha\geq0$. This shows that $f$ is convex on $(0,1]$ and by continuity on $[0,1]$ which completes the proof.
\end{proof}